\documentclass[11pt]{article}

\usepackage{fullpage}

\newcommand{\ER}{Erd\H{o}s-R\'{e}nyi }
\newcommand{\A}{\mathcal{A}}
\usepackage{comment}
\usepackage{epsf}
\usepackage{graphicx}
\usepackage{subfigure}
\usepackage{bbm}
\usepackage{color}
\usepackage{nicefrac}
\usepackage{mathtools}
\usepackage{amsfonts}
\usepackage{amsthm}
\usepackage{amsmath, bm}
\usepackage{amssymb}
\usepackage{textcomp}
\usepackage{xcolor}
\usepackage{algorithm}
\usepackage{algpseudocode}
\usepackage{comment}
\usepackage{enumitem}
\makeatletter
\algnewcommand{\LineComment}[1]{\Statex \hskip\ALG@thistlm {\color{gray}\texttt{// #1}}}
\makeatother

\usepackage{hyperref}
\definecolor{OliveGreen}{HTML}{3C8031}
\hypersetup{
colorlinks=true,
urlcolor=blue,
linkcolor=blue,
citecolor=OliveGreen,
linktoc=page,
}

\newenvironment{fminipage}%
  {\end{minipage}\end{Sbox}\fbox{\TheSbox}}

\allowdisplaybreaks

\makeatletter
\newenvironment{breakablealgorithm}
  {% \begin{breakablealgorithm}
   \begin{center}
     \refstepcounter{algorithm}% New algorithm
     \hrule height.8pt depth0pt \kern2pt% \@fs@pre for \@fs@ruled
     \renewcommand{\caption}[2][\relax]{% Make a new \caption
       {\raggedright\textbf{\ALG@name~\thealgorithm} ##2\par}%
       \ifx\relax##1\relax % #1 is \relax
         \addcontentsline{loa}{algorithm}{\protect\numberline{\thealgorithm}##2}%
       \else % #1 is not \relax
         \addcontentsline{loa}{algorithm}{\protect\numberline{\thealgorithm}##1}%
       \fi
       \kern2pt\hrule\kern2pt
     }
  }{% \end{breakablealgorithm}
     \kern2pt\hrule\relax% \@fs@post for \@fs@ruled
   \end{center}
  }
\makeatother

\newtheorem{theorem}{Theorem}[section]

\newtheorem{proposition}[theorem]{Proposition}
\newtheorem{lemma}[theorem]{Lemma}

\newtheorem{definition}[theorem]{Definition}
\newtheorem{remark}[theorem]{Remark}
\newtheorem{observation}[theorem]{Observation}
\newtheorem{conjecture}[theorem]{Conjecture}

\usepackage[backend=biber, style=alphabetic, sorting=nyt, maxnames=100,backref=true, sortcites = true]{biblatex}
\newlength{\widebarargwidth}
\newlength{\widebarargheight}
\newlength{\widebarargdepth}

\makeatletter
\long\def\@makecaption#1#2{
       \vskip 0.8ex
       \setbox\@tempboxa\hbox{\small {\bf #1:} #2}
       \parindent 1.5em 
       \dimen0=\hsize
       \advance\dimen0 by -3em
       \ifdim \wd\@tempboxa >\dimen0
               \hbox to \hsize{
                       \parindent 0em
                       \hfil 
                       \parbox{\dimen0}{\def\baselinestretch{0.96}\small
                               {\bf #1.} #2
                               } 
                       \hfil}
       \else \hbox to \hsize{\hfil \box\@tempboxa \hfil}
       \fi
       }
\makeatother

\long\def\comment#1{}

\newcommand{\ind}[1]{\ensuremath{{\mathbb{I}\left\{ #1 \right\}}}}
\newcommand{\CProb}[2]{\mathbb{P}\left[#1 \mid #2\right]}
\newcommand{\BER}{\ensuremath{\mbox{\sf Ber}}}
\newcommand{\BIN}{\ensuremath{\mbox{\sf Bin}}}

\newcommand{\set}[1]{\{#1\}}

\newcommand{\N}{\mathbb{N}}
\newcommand{\Z}{\mathbb{Z}}

\newcommand{\E}{\mathbb{E}}

\newcommand{\ERB}{G_{\mathrm{bip}}}

\newcommand\numberthis{\addtocounter{equation}{1}\tag{\theequation}}

\title{Sharp Online Hardness for Large Balanced Independent Sets}

\usepackage{authblk}

\author[1]{Abhishek Dhawan\thanks{Email: \textit{adhawan2@illinois.edu}. Partially supported by the NSF RTG grant DMS-1937241.}}

\author[2]{Eren C. K{\i}z{\i}lda\u g\thanks{Email: \textit{kizildag@illinois.edu}.}}

\author[3]{Neeladri Maitra\thanks{Email: \textit{nmaitra@illinois.edu}.}}

\affil[1,3]{Department of Mathematics, University of Illinois Urbana-Champaign}

\affil[2]{Department of Statistics, University of Illinois Urbana-Champaign}

\date{}

\begin{document}
\maketitle

\begin{abstract}
    We study the algorithmic problem of finding large $\gamma$-balanced independent sets in dense random bipartite graphs; an independent set is $\gamma$-balanced if a $\gamma$ proportion of its vertices lie on one side of the bipartition. In the sparse regime, Perkins and Wang~\cite{perkins2024hardness} established tight bounds within the low-degree polynomial (LDP) framework, showing a factor-$1/(1-\gamma)$ statistical–computational gap via the Overlap Gap Property (OGP) framework tailored for stable algorithms. However, these techniques do not appear to extend to the dense setting. For the related large independent set problem in dense Erd\H{o}s-R\'{e}nyi random graph $G(n,p)$, the best known algorithm is an online greedy procedure that is inherently unstable, and LDP algorithms are conjectured to fail even in the ``easy'' regime where greedy succeeds.

    For constant $p,\gamma \in (0, 1)$, we show that the largest $\gamma$-balanced independent set in $G_{\text{bip}}(n,p)$ has size $\alpha_{\rm STAT}:=\frac{\log_b n}{\gamma(1-\gamma)}$ with high probability (whp), where $n$ is the size of each bipartition, $p$ is the edge probability, and $b=1/(1-p)$. We design a two-stage online algorithm—revealing vertices sequentially and making irrevocable decisions based solely on current information—that achieves $(1-\epsilon)\alpha_{\rm COMP}$ whp for any $\epsilon>0$, where $\alpha_{\rm COMP}:=(1-\gamma)\alpha_{\rm STAT}$. We complement this with a sharp lower bound, showing that no online algorithm can achieve $(1+\epsilon)\alpha_{\rm COMP}$ with nonnegligible probability.
    
    Our results suggest that the same factor-$1/(1-\gamma)$ gap is also present in the dense setting, supporting its conjectured universality. While the classical greedy procedure on $G(n,p)$ is straightforward, our algorithm is more intricate: it proceeds in two stages, incorporating a stopping time and suitable truncation to ensure that $\gamma$-balancedness—a global constraint—is met despite operating with limited information. Our lower bound utilizes the OGP framework. Although the traditional scope of the OGP has been stable algorithms, we build on a recent refinement of this framework for online models and extend it to the bipartite setting.
\end{abstract}
\pagenumbering{gobble}
\newpage
\tableofcontents
\newpage
\pagenumbering{arabic}
%\setcounter{page}{1}
%%%%%%%

\sloppy

\section{Introduction}

In this paper, we study the algorithmic problem of finding large balanced independent sets in dense random bipartite graphs. While finding large independent sets—or even approximating them to within an $n^{1-\epsilon}$ factor—is NP-hard in the worst-case~\cite{hastad-clique,khot2001improved}, the situation becomes far more intriguing in the presence of randomness. 

For the \ER random graph $G(n,\frac12)$, the largest independent set has size approximately $2\log_2 n$ with high probability (whp)~\cite{matula1970complete,matula1976largest,Grimmett_McDiarmid_1975,bollobas1976cliques}. Moreover, a simple greedy algorithm operating in an online fashion—where vertices are revealed sequentially, and the decision at step $t$ depends only on partial information available at step $t$—finds an independent set of size $\log_2 n$ with high probability~\cite{Grimmett_McDiarmid_1975}. In 1976, Karp asked whether it is possible to design an efficient algorithm that, whp, finds an independent set of size $(1+\epsilon)\log_2 n$ for $\epsilon>0$~\cite{karp1976probabilistic}. Surprisingly, this question remains open and is widely believed to be computationally intractable. It is worth mentioning that proving the hardness of Karp’s task unconditionally would imply $P\ne NP$.

Karp’s problem stands as a central question in average-case complexity and the algorithmic theory of random graphs.\footnote{For instance, Frieze explicitly highlighted it as a major open problem in his 2014 ICM plenary lecture~\cite{F2014}.} It is perhaps the earliest instance of a  \emph{statistical-computational gap}—a gap between the existential bound and the best known polynomial-time algorithm. This gap has been extensively studied (e.g., prompting Jerrum to propose the now-famous planted clique model~\cite{jerrum1992large}), and a large body of work has since uncovered similar ``factor 2-gaps” in other random graph models, suggesting a certain universality. For a broad overview of such gaps both in the context of random graphs and beyond, see the surveys~\cite{bandeira2018notes,gamarnik2021overlap,gamarnik2022disordered,gamarnik2025turing}.

The sparse case, $G(n,\frac{d}{n})$ with constant $d$, has seen substantial progress. In this case, the largest independent set has size $\sim 2n\frac{\log d}{d}$~\cite{frieze1990independence,frieze1992independence,bayati2010combinatorial}, while the best known efficient algorithm only achieves $\sim n\frac{\log d}{d}$, both whp~\cite{lauer2007large}.\footnote{Both guarantees hold in the double limit $n\to\infty$ followed by $d\to\infty$.} Is it hard to find independent sets of size $(1+\epsilon)n\frac{\log d}{d}$? As noted above, resolving this question would imply a conclusion even stronger than $P \ne NP$. Accordingly, contemporary research has instead focused on providing rigorous evidence of hardness, e.g., by establishing unconditional lower bounds for certain classes of algorithms. This algorithmic gap led Gamarnik and Sudan to introduce the \emph{Overlap Gap Property} (OGP) framework~\cite{gamarnik2014limits}. Leveraging this framework, tight hardness results were obtained for powerful algorithmic classes, including low-degree polynomial (LDP) algorithms~\cite{wein2020optimal} and local algorithms~\cite{rahman2017local}. These arguments were subsequently extended to sparse random bipartite graphs by Perkins and Wang~\cite{perkins2024hardness}, a work closely related to ours (see below); see also~\cite{dhawan2024low} for an extension to hypergraphs.

\paragraph{Dense Random Graphs} The situation is markedly different in the dense regime, $G(n,p)$ with constant $p$. In this setting, while the hardness results were recently established for LDP algorithms~\cite{huang2025strong}, it remains unknown whether the greedy algorithm can itself be implemented as an LDP. In fact, it has been conjectured in a recent AIM workshop~\cite{AIM2024} that for $G(n,\frac12)$:
\begin{conjecture}\label{conj:LDP}
   No degree-$o(\log^2 n)$ polynomial can return an independent set of size $0.9\log_2 n$.
\end{conjecture}
That is, LDP algorithms likely fail even in the regime where the greedy algorithm succeeds, suggesting they may not be viable at all in the dense setting. Given that LDP algorithms are quite powerful in many high-dimensional problems, this is particularly surprising—especially since the greedy algorithm operates using only partial, sequential information without accessing the full graph. Thus, different techniques are needed for analyzing the greedy algorithm and the online setting.
  
A recent work~\cite{gamarnik2025optimal} has refined the OGP framework and subsequently obtained sharp lower bounds for a broad class of online algorithms, which includes the greedy algorithm as a special case. Our work extends these techniques to dense random bipartite graphs, as we detail below.

\subsection{Random Bipartite Graphs}
Bipartite graphs arise frequently in modeling real-world scenarios with an inherent two-part structure
(e.g., job assignment). From a theoretical standpoint, many classical problems in graph theory and extremal combinatorics—such as Turán- and Ramsey-type questions—have natural bipartite counterparts. In our context, random bipartite graphs are a natural testbed for investigating the robustness of statistical-computational gaps (e.g., the factor-$2$ gap discussed above).

Our focus is on the dense \ER random bipartite graph $\ERB(n, p)$ defined on disjoint vertex sets $(L,R)$ where $|L|=|R|=n$ and $p$ is a constant. Each edge between $L$ and $R$ is included independently with probability $p$. 
In general, maximum independent sets in bipartite graphs can be found efficiently via max-flow. Likewise, finding large independent sets in $\ERB(n,p)$ is algorithmically easy; in fact, even approximately counting or sampling such sets in the sparse case $G_{\text{bip}}(n, d/n)$ is tractable—see~\cite{dhawan2024low} and references therein. However, powerful  statistical mechanics heuristics~\cite{mezard1987mean} suggest that introducing global constraints may lead to a \emph{glassy} phase and computational hardness. 

A natural global constraint, as studied in~\cite{perkins2024hardness}, is to consider \emph{balanced} independent sets $I$, where $|I\cap L|=|I\cap R|$. More generally, they allow for a specified proportion of vertices from each side of the bipartition, which is also our focus. Indeed, introducing such a global constraint already leads to computational barriers: finding a largest balanced independent set in a bipartite graph is NP-hard~\cite{gareyjohnson,feige2002relations}. For further references and detailed discussion, see~\cite{perkins2024hardness}.

In this paper, we study $\gamma$-balanced independent sets—those in which a $\gamma$ fraction of vertices lie in one part and a $(1 - \gamma)$ fraction in the other, where $\gamma\in (0,\frac12]$ without loss of generality.
\begin{definition} %\textcolor{red}{{\bf Check\& modify}} 
Given a bipartite graph $G$ with bipartition $(L,R)$, an independent set $I$ of $G$ is $\gamma$-balanced if $\bigl||I\cap L|-\gamma |I| \bigr|<1$ or $\bigl||I\cap R|-\gamma |I| \bigr|<1$.
\end{definition}

For the sparse case, $\ERB(n, d/n)$, Perkins and Wang~\cite{perkins2024hardness} provide a fairly comprehensive picture. They show that (i) the largest $\gamma$-balanced independent set has size $\left(1/(\gamma(1-\gamma))\pm o_d(1)\right)n\frac{\log d}{d}$, (ii) there are local/LDP algorithms finding an independent set of size $\left(1/\gamma-o_d(1)\right)n\frac{\log d}{d}$, and (iii) for $d$ large, local/LDP algorithms fail to find an independent set of size $\left(1/\gamma+\epsilon\right)n\frac{\log d}{d}$.  That is, a factor-$1/(1 - \gamma)$~statistical-computational gap emerges. In the balanced case $\gamma = 1/2$, this reproduces the familiar factor-$2$ gap. 

Can we extend these results to dense bipartite graphs? 
To begin with, their algorithmic result crucially exploits sparsity—specifically, the fact that sparse graphs are locally `tree-like'. This structural property no longer holds in the dense regime. In fact, as discussed earlier and suggested by Conjecture~\ref{conj:LDP}, such algorithms may not even be viable candidates in the dense setting. Instead, online algorithms emerge as a natural candidate. However, enforcing the $\gamma$-balancedness constraint in an online setting with randomly ordered vertex arrivals poses new difficulties (see below). Furthermore, the core of their hardness result relies on the stability of local/LDP algorithms. Online algorithms, however, may be unstable; see~\cite{gamarnik2025optimal}.

\paragraph{Dense Bipartite Graphs} In this paper, we characterize the landscape for dense random bipartite graphs, including (i) the statistical threshold $\alpha_{\rm STAT}$ for the largest $\gamma$-balanced independent set, (ii) an online algorithm achieving the threshold $\alpha_{\rm COMP}=(1-\gamma)\alpha_{\rm STAT}$, and (iii) a sharp algorithmic lower bound, showing that no online algorithm can surpass $\alpha_{\rm COMP}$. This suggests that $\alpha_{\rm COMP}$ is the computational threshold for this problem. In particular, our results establish a similar factor-$1/(1-\gamma)$ statistical-computational gap, thus reinforcing its apparent universality.

Before describing our main results, we highlight key challenges. First and foremost, any approach that incurs even mild logarithmic factors would fall short of addressing a constant-factor computational gap. In the sparse setting, the OGP framework has proven to be a powerful tool for addressing constant-factor gaps. However, this framework applies primarily to stable algorithms, including local/LDP algorithms. In light of Conjecture~\ref{conj:LDP}, such algorithms are unlikely to be viable in the dense case. Instead, online algorithms offer a more natural starting point—further motivated by the fact that the best known algorithm for $G(n,p)$ itself is an online algorithm. Despite this, applications of the OGP framework to online algorithms remain only a handful, see Section~\ref{subsection: OGP}. Even for classical greedy on $G(n,p)$, lower bounds via OGP require delicate refinements and have been established only very recently~\cite{gamarnik2025optimal}. In our setting, the technical challenges are compounded by two additional factors: (i) the greedy algorithm operates with only partial information, while the $\gamma$-balancedness constraint is global; and (ii) the random vertex arrival order—a core feature of online algorithms—can lead to situations revealing few or no cross-edges in the bipartite graph, further complicating the analysis (see below).

In the classical online setting (e.g., greedy on $G(n, p)$), vertices arrive sequentially and in random order. At each step $t$, the algorithm decides whether to include the incoming vertex $v_t$ by inspecting its connections to $I_{t-1}$, the independent set constructed thus far. Crucially, the decision for $v_t$ must be made immediately—it cannot be deferred. 

The introduction of the $\gamma$-balancedness constraint makes the situation more delicate. In the setting of~\cite{perkins2024hardness}, the algorithm is not online; once an independent set is constructed, it can be rebalanced—e.g., by discarding vertices—to satisfy the balancedness constraint. In our setting, however, such post-hoc pruning is not allowed, as it breaks the requirement that decisions are irrevocable. Consequently, the algorithm must be cognizant of the balancedness constraint throughout its execution.

This challenge is further compounded by the randomness in vertex arrival order. Crucially, if majority of the vertices come from the same side of the bipartition in, e.g., the first $n$ steps, the algorithm receives limited information about the status of cross-edges. In contrast, in $G(n,p)$, each step reveals the status of some new edges, regardless of the vertex order. 

The random arrival order, along with the $\gamma$-balancedness constraint also necessitates truncation, elaborated next. Consider the case where the first $n$ (or $n-o(n)$) vertices all come from the same side of the bipartition. A greedy algorithm might add all of them—clearly violating the balancedness constraint. More subtly, even adding as few as $(1+\epsilon)\log_b n$ of these vertices can prevent any subsequent vertex from the opposite side from being included—indicating that a careful truncation (without violating onlineness) is essential. Hence, while the analysis for classical greedy on $G(n,p)$ is rather easy, this is not the case in our setting. Moreover, any meaningful algorithmic lower bound should be oblivious to the arrival order—including adversarial ones such as the scenarios described above.

To address these challenges, we construct suitable auxiliary stochastic processes to track the evolution of the algorithm’s output, along with judiciously chosen stopping times that enforce the balancedness constraint without violating the online requirement.

\subsection{Summary of Main Results}
Recall that $\ERB(n,p)$ is the random bipartite graph on vertex set $(L,R)$ with $|L|=|R|=n$, where each edge in $L\times R$ is present independently with probability $p=\Theta(1)$. Fix $\gamma\in(0,\frac12]$ and set
\begin{equation}\label{eq:a-star}
    \alpha_{\rm STAT}\coloneqq\frac{\log_b n}{\gamma(1-\gamma)},\quad\quad \text{where}\quad\quad b\coloneqq b(p)=\frac{1}{1-p}.
\end{equation}
For the remainder of this paper, we ignore all floor/ceiling operators for simplicity with the understanding that this does not affect the overall arguments. 
\begin{theorem}[Informal, see Theorem~\ref{thm:Max-IS}]\label{thm:largest-IS-informal}
The largest $\gamma$-balanced independent set in $\ERB(n,p)$ has size $\bigl(1\pm o(1)\bigr)\alpha_{\rm STAT}$ whp.
\end{theorem}
Theorem~\ref{thm:largest-IS-informal} identifies $\alpha_{\rm STAT}$ in~\eqref{eq:a-star} as the \emph{statistical threshold}. Equipped with this, a natural algorithmic question arises: can we efficiently find large $\gamma$-balanced independent sets? We set
\begin{equation}\label{eq:alpha-alpg}
    \alpha_{\rm COMP} \coloneqq (1-\gamma)\alpha_{\rm STAT} = \frac{\log_b n}{\gamma}.
\end{equation}
Our next result shows that $\alpha_{\rm COMP}$ is attainable in polynomial time.
\begin{theorem}[Informal, see Theorem~\ref{theorem: achievability}]\label{thm:alg-informal}
There is an online algorithm—oblivious to the vertex arrival order—which, for any $\epsilon > 0$, finds a $\gamma$-balanced independent set of size $(1 - \epsilon)\alpha_{\mathrm{COMP}}$ whp.
\end{theorem}
See Definition~\ref{def:OnlineAlg} for a description of online algorithms. The decision at time $t$ takes polynomial time, so the overall runtime of our algorithm is polynomial in $n$.
Notably, our algorithm is completely oblivious to the vertex arrival order. As mentioned earlier, while the analysis of the classical greedy algorithm on $G(n,p)$ is relatively straightforward, the situation here is more delicate. The presence of the $\gamma$-balancedness constraint, along with the random arrival order, introduces additional challenges. We address these by (i) analyzing suitable stochastic processes that track the algorithm's evolution, and (ii) employing a careful truncation to ensure the $\gamma$-balancedness.

Observe that there is a factor-$1/(1-\gamma)$ gap between the statistical threshold and the algorithmic value, reminiscent of the computational gap in the sparse setting. Our final result addresses this gap and establishes a sharp computational lower bound.
\begin{theorem}[Informal, see Theorem~\ref{theorem: impossibility}]\label{thm:OGP-informal}
For any $\epsilon>0$, no online algorithm finds a $\gamma$-balanced independent set of size $(1+\epsilon)\alpha_{\rm COMP}$ with probability at least $\exp(-O(\log^2 n))$.
\end{theorem}

We note that there is no restriction on the runtime of the algorithms ruled out. Furthermore, Theorem~\ref{thm:OGP-informal} establishes strong hardness—it rules out online algorithms that succeed even with vanishing probability. Importantly, the probability guarantee is essentially optimal: the probability that a randomly chosen $\gamma$-balanced set of size $(1+\epsilon)\alpha_{\rm COMP}$ is an independent set is itself $\exp(-\Theta(\log^2 n))$.

Taken together, Theorems~\ref{thm:alg-informal} and~\ref{thm:OGP-informal} provide a tight characterization of the performance of online algorithms, providing rigorous evidence toward the following conjecture (we note that a version of this conjecture in the sparse regime appears in \cite{perkins2024hardness}).
\begin{conjecture}
    For any $\epsilon>0$ and $p = \Theta(1)$, no polynomial-time algorithm finds a $\gamma$-balanced independent set of size $(1+\epsilon)\alpha_{\rm COMP}$ in $\ERB(n, p)$ whp.
\end{conjecture}
\paragraph{Future Queries}
The setting of~\cite{gamarnik2025optimal} permits querying a limited set of \emph{future edges}—edges incident to vertices not yet seen—at each step. That is, the decision at time $t$ is based not only on the edges revealed so far, but also on a restricted set of such future edges. In this augmented model,~\cite{gamarnik2025optimal} prove both lower bounds and algorithmic guarantees: specifically, they show that algorithms with modest access to future information can in fact exceed the $(1+\epsilon)\log_b n$ threshold, albeit using quasi-polynomial time.

This naturally raises the following question: can algorithms with limited future information outperform the $(1+\epsilon)\alpha_{\rm COMP}$ threshold for $\ERB(n,p)$? We show that the answer is yes:
\begin{theorem}[Informal, see Theorem~\ref{thm:surpass}]\label{thm:surpass-informal}
For any $\epsilon>0$, there exists an online algorithm which makes limited future queries and finds a $\gamma$-balanced independent set of size $(1+\epsilon)\alpha_{\rm COMP}$ with probability at least $1-\exp(-n^{\Theta(1)})$.
\end{theorem}
Our algorithm is online, and runs in super-polynomial time. It has a mild dependence on the vertex arrival order—indeed, no online algorithm that is fully oblivious to arrival order can surpass $(1+\epsilon)\alpha_{\rm COMP}$. See Section~\ref{sec:future-sec} for further discussion.

\subsection{Proof Overview}

In this section, we will provide an overview of our proof techniques.
The proof of the statistical threshold follows a standard application of the first and the second moment similar to that of $G(n, p)$ tailored to the bipartite setting.
Much of our effort is in the proof of the computational threshold.

\paragraph{Achievability result}
The greedy algorithm for ordinary independent sets is inherently online, however, as mentioned earlier the global nature of the $\gamma$-balancedness constraint introduces technical challenges.
We overcome these challenges by designing a two-stage online algorithm with truncation steps, i.e., we stop adding vertices to the independent set $I$ from a partition $\eta$ once we have reached the desired number of vertices in $I \cap \eta$.
More formally:
\begin{enumerate}
    \item In stage one, we greedily add vertices to the independent set $I$ until $I$ contains $(1-\epsilon)\gamma\alpha_{\rm COMP}$ vertices from one partition.
    \item At this point, without loss of generality, let us assume $|I \cap L| = (1-\epsilon)\gamma\alpha_{\rm COMP}$.
    During stage two, we only add vertices in $R$ to the independent set, stopping once $|I \cap R| = (1-\epsilon)(1 - \gamma)\alpha_{\rm COMP}$.
\end{enumerate}
Note that there are two ``truncation'' points in our algorithm (one in each stage).
Furthermore, since $\gamma \leq 1/2$ we may assume $|I \cap R| < (1-\epsilon)(1 - \gamma)\alpha_{\rm COMP}$ at the beginning of stage two.

In order to prove our acheivability result, we must show that enough vertices from $R$ are added to $I$ during stage two whp.
Under the assumption that $|I \cap L| = (1-\epsilon)\gamma\alpha_{\rm COMP}$, the probability that an arbitrary vertex in $R$ considered during stage two is in fact added to $I$ is $n^{-1 + \epsilon}$.
Therefore, it is enough to show that $\gg n^{1-\epsilon}$ vertices in $R$ remain to be exposed during stage two.
The key part of our analysis, therefore, is the following result: \textit{stage one concludes in at most $n^{1-\epsilon/2}$ steps whp irrespective of the vertex arrival order} (Lemma~\ref{lemma: greedy analysis}).

\paragraph{Impossibility result}
The proof of the upper bound of our computational threshold falls within the OGP framework (see Section~\ref{subsection: OGP} for a brief history of the technique).
As mentioned earlier, OGP-based arguments have predominantly served as a barrier to stable
algorithms.
As online algorithms may not be stable, one needs to refine the approach to adapt it to this setting.
Our work is one of the first to do so.

Given an algorithm $\mathcal{A}$, we aim to bound the probability, denoted by $\delta$, that $\mathcal{A}$ finds a $\gamma$-balanced independent set of size at least $(1+\epsilon)\alpha_{\rm COMP}$.
At the heart of our proof lies a sequence of correlated random graphs $(G_{i}^{(T)})_{i\in [m], T \in [2n]}$ for $m = \Theta(\epsilon^{-2})$.
We define a \textit{successful event $\mathcal{S}$} determined by running the algorithm $\mathcal{A}$ on each of the graphs $G_i^{(T)}$.
Roughly speaking, $\mathcal{S}$ denotes the event that for a specific timestep $\tau \coloneqq \tau\left(\mathcal{A},\, (G_{i}^{(T)})_{i\in [m], T \in [2n]}\right)$, the independent sets $\mathcal{A}(G_1^{(\tau)}), \ldots, \mathcal{A}(G_m^{(\tau)})$ all have size at least $(1+\epsilon)\alpha_{\rm COMP}$.
We then show the following:
\begin{enumerate}
    \item $\mathbb{P}[\mathcal{S}] \geq \delta^m$, and
    \item $\mathbb{P}[\mathcal{S}] = \exp(-\Omega(\log_b^2n))$.
\end{enumerate}
Combining the above completes the proof.

We note that the sequence of correlated graphs $(G_{i}^{(T)})_{i\in [m], T \in [2n]}$ is defined with respect to the algorithm $\mathcal{A}$.
In particular, the vertex arrival order determines the correlations within the sequence.
This is in stark contrast to other applications of the OGP framework and is the key refinement to adapt the OGP technique to the (potentially unstable) online setting.
For instance, in \cite{wein2020optimal, perkins2024hardness, dhawan2024low}, the authors define sequences of correlated random (hyper)graphs independent of the algorithm $\mathcal{A}$ in question.
They then apply the algorithm to each graph to construct a \textit{forbidden substructure} whp—a sequence of independent sets which appear with low probability; therefore, arriving at a contradiction.

\subsection{Open Problems}

We conclude this introduction with a description of potential future directions of inquiry.

\paragraph{One-Stage Algorithm} Our algorithm relies on a two-stage structure to enforce the $\gamma$-balancedness constraint. A natural question is whether a one-stage algorithm—akin to greedy on $G(n,p)$—can achieve $\alpha_{\rm COMP}$. This seems difficult due to the $\gamma$-balancedness constraint, though one might try introducing a bias at each step: based on the vertex’s side of the bipartition, current balance, and connectivity, decide via a coin flip. We leave this for future work.

\paragraph{Hypergraphs} A recent work~\cite{dhawan2024low} extends the results of~\cite{gamarnik2014limits,rahman2017local,wein2020optimal} to sparse random $r$-uniform hypergraphs, obtaining algorithmic guarantees and sharp computational lower bounds for LDP algorithms, and demonstrating the emergence of an analogous statistical–computational gap. It also investigates the universality of this gap via a multipartite hypergraph version of the largest balanced independent set problem (introduced in \cite{dhawan2025balanced}) in sparse $r$-uniform $r$-partite hypergraphs, recovering and generalizing the results of~\cite{perkins2024hardness}. Extending our results and the results of \cite{gamarnik2025optimal} to dense random hypergraphs are both interesting directions for future work.

\paragraph{Optimization Problems with Global Constraints}
The balanced independent set problem in bipartite graphs is an example of how a problem in $P$ can be made $NP$-hard by imposing global constraints.
It is worth investigating whether such gaps persist across similar problems.
For instance, the largest induced matching problem—finding the largest matching $M \subseteq E(G)$ such that $E(G[V(M)]) = M$—is $NP$-hard \cite{cameron1989induced}.
The statistical threshold was determined in \cite{cooley2021large}, while the computational threshold is unknown.
Another problem to explore would be the $m$-partite graph analogue of the largest balanced independent set problem; we note that a coloring variant of this problem was suggested in \cite{chakraborti2023extremal} for deterministic graphs.

\section{Statistical-Computational Gaps and OGP: Prior Work}
A wide range of problems across probability theory, computer science, high-dimensional statistics, and machine learning involve randomness and exhibit a common phenomenon: a statistical-computational gap. That is, there is often a discrepancy between what is information-theoretically possible and what is achievable by efficient algorithms. For example, in random optimization problems such as the one we study, the optimal value can often be identified through non-constructive means. However, known polynomial-time algorithms yield strictly suboptimal solutions, and no efficient method is known for finding a global optimum without brute-force search. The models with such an apparent gap include random constraint satisfaction problems (CSP)~\cite{mezard2005clustering,achlioptas2006solution,achlioptas2008algorithmic,gamarnik2017performance,bresler2021algorithmic,yung2024}, spin glass models~\cite{chen2019suboptimality,huang2021tight,huang2023algorithmic,gamarnikjagannath2021overlap,gamarnik2020low,gamarnik2023shattering,kizildaug2023sharp,sellke2025tight}, number balancing and discrepancy minimization~\cite{mallarapu2025strong,gamarnik2023algorithmic,gamarnik2023geometric}, Ising perceptron~\cite{gamarnik2022algorithms,li2024discrepancy}, as well as various computational problems over random graphs~\cite{gamarnik2014limits,gamarnik2017,gamarnik2020low,wein2020optimal,perkins2024hardness,ding2023low,dhawan2024low} and more.

Standard complexity theory is tailored primarily to worst-case hardness and offers limited insight into average-case models (see~\cite{ajtai1996generating,boix2021average,GK-SK-AAP,vafa2025symmetric} for a few notable exceptions). Nevertheless, these gaps are a very active area of investigation; researchers have developed various frameworks for providing rigorous evidence of hardness. For an overview of these methods, we refer the reader to the excellent surveys~\cite{wu2018statistical,bandeira2018notes,gamarnik2021overlap,gamarnik2022disordered,gamarnik2025turing,wein2025computational}. 
\subsection{Computational Gaps in Random Optimization Problems}\label{subsection: OGP}
For random optimization problems, arguably the most powerful framework for establishing algorithmic hardness is the Overlap Gap Property (OGP) introduced by Gamarnik and Sudan~\cite{gamarnik2014limits} (and formally named in~\cite{gamarnik2018finding}). Building on insight from statistical physics—particularly the intriguing connection between the onset of algorithmic hardness and geometric phase transitions in random CSPs~\cite{mezard2005clustering,achlioptas2008algorithmic,achlioptas2006solution}—the OGP framework has proven instrumental in establishing rigorous algorithmic barriers by leveraging intricate geometry of the optimization landscape. For surveys on OGP, see~\cite{gamarnik2021overlap,gamarnik2025turing}.

We briefly describe this framework in its original context: finding large independent sets in sparse random graphs. Gamarnik and Sudan~\cite{gamarnik2014limits} established that independent sets of size $(1+1/\sqrt{2})n\frac{\log d}{d}$ exhibit an `overlap-gap': any two such sets have either a large or a small intersection, with no overlaps of intermediate size. This structural property enabled them to rule out local algorithms at this threshold, thereby refuting a conjecture
of Hatami-Lov\'asz-Szegedy~\cite{hatami2014limits} which posited that local algorithms can find maximum independent sets in $d$-regular random graphs. 

Subsequent work by Rahman-Vir\'ag~\cite{rahman2017local} extended this hardness result down to the sharp threshold $n\frac{\log d}{d}$, below which polynomial-time algorithms are known~\cite{lauer2007large}. Unlike the pairwise OGP considered in earlier work, their approach relied on analyzing overlaps among multiple independent sets—a notion termed multi-OGP—to obtain tight lower bounds. Recent works have introduced refined variants of the multi OGP: for instance, asymmetric versions of OGP yielded tight lower bounds against LDP algorithms~\cite{wein2020optimal,bresler2021algorithmic}, and the branching OGP~\cite{huang2021tight} has emerged as a very powerful tool in the study of spin glasses. 
The OGP framework has since become the `bread-and-butter' for proving sharp computational lower bounds in numerous random optimization problems. For certain models—such as Ising perceptron and discrepancy minimization—OGP-based hardness results are complemented by more traditional notions of average-case hardness (e.g., worst-case hardness of approximating the shortest vector in lattices)~\cite{vafa2025symmetric}.\footnote{Interestingly, there exist models exhibiting the OGP, which remain solvable in polynomial time (e.g., by linear programming)~\cite{li2024some}—beyond the classical counterexample of random XOR-SAT solvable by Gaussian elimination.} The literature on OGP is now quite extensive; we refer the reader to references above.
\paragraph{OGP for Online Algorithms} The OGP framework is primarily tailored for stable algorithms—those whose output is insensitive to small variations in the input.\footnote{Informally, an algorithm $\A$ is stable if for any inputs $G,G'$ with small $\|G-G'\|$, the outputs $\A(G),\A(G')$ are close.} Many prominent algorithms for average-case models fall into this category, including local algorithms (e.g., factors of iid)~\cite{gamarnik2014limits,rahman2017local}, LDP algorithms, approximate message passing~\cite{gamarnikjagannath2021overlap}, Boolean circuits with low depth~\cite{gamarnik2021circuit}, as well as gradient descent and Langevin dynamics~\cite{gamarnik2020low}. OGP-based hardness arguments rely critically on this stability—for instance, to construct interpolation arguments showing that the algorithm's trajectory evolves smoothly and thus avoids the intermediate overlap region which is forbidden in the solution space. Online algorithms, however, may be unstable—see~\cite[Proposition~1.1]{gamarnik2025optimal}.

Can OGP-based barriers be extended to online algorithms? This question is especially relevant in the modern era of big data, where the online setting is a natural model of decision-making under uncertainty and have been extensively studied in the optimization and machine learning literature~\cite{rakhlin2010online,rakhlin2011online-a,rakhlin2011online-b,rakhlin2013online,hazan2016introduction}. This question has first been addressed in~\cite{gamarnik2023geometric}, where sharp lower bounds for online algorithms were obtained for the Ising perceptron. More recently, OGP-based barriers for online algorithms were obtained for the graph alignment problem~\cite{du2025algorithmic} and the largest submatrix problem~\cite{bhamidi2025finding}. 

Extending such online barriers to random graphs—the very setting where the OGP has first emerged in—turned out to be quite challenging. This is particularly due to the lack of stability, a key feature that OGP-based arguments crucially build on. The first lower bounds for online algorithms in $G(n,p)$ were obtained in~\cite{gamarnik2025optimal} through novel technical refinements. Their arguments include (i) the construction of temporal interpolation paths that evolve with the algorithm (in contrast to earlier OGP-based barriers, which are algorithm-independent) as well as (ii) the use of stopping times tracking the size of the output. In the present paper, we extend the techniques of~\cite{gamarnik2025optimal} to the bipartite setting where the arguments are further refined to (i) incorporate the $\gamma$-balancedness constraint, and (ii) handle the random arrival order, which may reveal very few cross-edges—thus providing limited information. 
\section{Main Results}
We begin by determining the size of the largest $\gamma$-balanced independent set in $\ERB(n,p)$ for constant $p$. Recall from~\eqref{eq:a-star} that $\alpha_{\rm STAT}=\frac{\log_b n}{\gamma(1-\gamma)}$, where $b=1/(1-p)$.

\begin{theorem}\label{thm:Max-IS}
Let  $Z_\alpha(\gamma)$ denote the number of $\gamma$-balanced independent sets of size $\alpha$. For any fixed $\epsilon\in(0,1)$ and $\alpha_{\rm STAT}$ as in~\eqref{eq:a-star}, the following hold:
\begin{enumerate}[label = (S\arabic*)]
    \item\label{stat: FMM} For $\alpha\ge (1+\epsilon)\alpha_{\rm STAT}$, $\mathbb{P}[Z_\alpha(\gamma)\ge 1] = \exp(-\Theta(\log^2 n))$.
    \item\label{stat: SMM} For $\alpha \le (1-\epsilon)\alpha_{\rm STAT}$, $\mathbb{P}[Z_\alpha(\gamma)\ge 1]\ge 1-\exp(-\Theta(\log n))$.
\end{enumerate}
\end{theorem}
Thus, the largest $\gamma$-balanced independent set is approximately of size $\alpha_{\rm STAT}$, which we refer to as the \emph{statistical threshold}. Theorem~\ref{thm:Max-IS} follows from a standard application of the first and the second moment method, see Section~\ref{section: stat thresh proof} for the proof.

Given this benchmark, a natural algorithmic question arises: can we find such independent sets efficiently? Motivated by the fact that the best known algorithm for the maximum independent set problem in $G(n,p)$ is an online greedy algorithm, we naturally investigate the performance of online algorithms. 
\subsection{Algorithmic Setting}
The class of \emph{online algorithms} we consider is formalized as follows.
\begin{definition}\label{def:OnlineAlg}
Let $G\sim \ERB(n,p)$ have vertex set $L\cup R$, where $|L|=|R|=n$ and $L\cap R=\varnothing$. A randomized algorithm $\mathcal{A}$ with internal randomness determined by seed $\omega$ runs for $2n$ rounds and keeps track of sets $L_t \subseteq L$ and $R_t \subseteq R$ (initially $L_0 = R_0 = \varnothing$). At each round $t \in [2n]$:
\begin{enumerate}
\item Based on $\omega$ and all information revealed so far, $\mathcal{A}$ randomly selects a vertex $v_t\in (L\cup R)\setminus (L_t\cup R_t)$ and reveals the status of all edges $(v_t,v)$, where $v\in L_t\cup R_t$. 
\item Based on $\omega$ and all information revealed so far, $\mathcal{A}$ then decides if $\mathcal{A}_t(G) = \mathcal{A}_{t-1}(G)\cup \{v_t\}$ and updates the sets: (i) $L_{t+1}=L_t\cup \{v_t\}$ if $v_t\in L$ or (ii) $R_{t+1}=R_t\cup \{v_t\}$ if $v_t\in R$. 
\end{enumerate}
\end{definition}
Per Definition~\ref{def:OnlineAlg}, the vertex arrival order is random, determined jointly by the algorithm’s internal randomness $\omega$ and the randomness of $G$ (its edges). If $v_t\in L$ (resp.\,$R$), then all edges from $v_t$ to vertices in $L$ (resp. $R$) are absent, so information arises only from edges to the opposite side of bipartition inspected so far. The algorithm may select multiple vertices from the same side in succession, potentially revealing no new information—for example, inspecting only $L$ in the first $n$ rounds and yielding $R_n = \varnothing$.

Our results hold for the most general setting: (i) the algorithmic bound (Theorem~\ref{theorem: achievability}) is independent of the arrival order, and (ii) the hardness result (Theorem~\ref{theorem: impossibility}) applies to all online arrival scenarios allowed by Definition~\ref{def:OnlineAlg}.

Our focus is on online algorithms that return large independent sets with specificed probability, formalized as follows.
\begin{definition} \label{Def:f_Optimize_ind}
    For parameters $k > 0$ and $\delta \in [0,1]$, an online algorithm $\mathcal{A}$ operating according to Definition~\ref{def:OnlineAlg} is said to $(k, \delta)$-optimize the $\gamma$-balanced independent set problem in $G_{\mathsf{bip}}(n, p)$ if the following is satisfied when $G \sim G_{\mathsf{bip}}(n, p)$:
    \[\mathbb{P}[|\mathcal{A}(G)| \geq k] \geq \delta.\]
\end{definition}

%\subsection{Statistical Threshold}
\subsection{Algorithmic Results}
Equipped with Definitions~\ref{def:OnlineAlg} and~\ref{Def:f_Optimize_ind}, we now present our algorithmic results. Recall $\alpha_{\rm COMP}=(1-\gamma)\alpha_{\rm STAT} = \log_b n/\gamma$ from~\eqref{eq:alpha-alpg}. Our first result shows $(1-\epsilon)\alpha_{\rm COMP}$ is achievable for any $\epsilon>0$.

\begin{theorem}\label{theorem: achievability}
    For any $\epsilon > 0$ and $p, \gamma \in (0, 1)$, there is an online algorithm $\mathcal{A}$ that $(k, \delta)$-optimizes the $\gamma$-balanced independent set problem in $G_{\mathsf{bip}}(n, p)$, where
    \[k = (1-\epsilon)\alpha_{\rm COMP}, \qquad \text{and} \qquad \delta = 1 - \exp\left(-\Omega\left(n^{\epsilon/2}\right)\right).\]
    % where $b = \frac{1}{1-p}$.
\end{theorem}
See Section~\ref{sec:achievability} for the proof. As noted earlier, the algorithm achieving $(1-\epsilon)\alpha_{\rm COMP}$ is online, implemented in two stages. In contrast, the classical greedy algorithm on $G(n,p)$ yields a straightforward online algorithm. The two-stage structure ensures that the global $\gamma$-balancedness constraint is satisfied by the final output, even though the algorithm itself operates using only local information.

We next complement Theorem~\ref{theorem: achievability} with a sharp lower bound.
\begin{theorem}\label{theorem: impossibility}
    For any $\epsilon>0$, 
    there exists no online algorithm that $(k, \delta)$-optimizes the balanced independent set problem in $G_{\mathsf{bip}}(n, p)$, where
    \[k = (1+\epsilon)\alpha_{\rm COMP}, \qquad \text{and} \qquad \delta = \exp\left(-O\left(\epsilon^2\log_b^2n\right)\right).\]
\end{theorem}
We prove Theorem~\ref{theorem: impossibility} through a refined version of the OGP framework adapted to the online setting, which leverages geometric properties of tuples of large $\gamma$-balanced independent sets. See Section~\ref{sec:impossibility} for the details.

Taken together, Theorems~\ref{theorem: achievability} and~\ref{theorem: impossibility} indicate the presence of factor-$1/(1-\gamma)$ statistical-computational gap with respect to online algorithms, with $\alpha_{\rm COMP}$ serving as the computational threshold for this model. As noted earlier, the same factor gap also appears in the context of sparse random bipartite graphs and LDP algorithms~\cite{perkins2024hardness}, further supporting the universality of this gap. 
%Notably, there is a factor-$\frac{1}{1-\gamma}$ statistical--computational gap as a result of Theorem~\ref{thm:Max-IS}. As discussed earlier, the same factor-$(1-\gamma)$ gap also arises in the context of sparse random bipartite graphs~\cite{perkins2024hardness}, which we believe is universal.
\begin{remark}\label{remark:prob-guarantee}
Observe that while there exists an algorithm succeeding modulo an exponentially small probability below $\alpha_{\rm COMP}$, even those with success probability $o(1)$ break down above $\alpha_{\rm COMP}$. This is known as strong hardness, see~\cite{huang2025strong}. We highlight that the probability guarantee in Theorem~\ref{theorem: impossibility} is essentially the best possible: the probability that a randomly selected, $\gamma$-balanced set of size $(1+\epsilon)\alpha_{\rm COMP}$ is an independent set is at most $\exp(-c\log_b^2 n)$ for some constant $c>0$.
\end{remark}

\subsection{Surpassing $\alpha_{\rm COMP}$ with Limited Future Queries}\label{sec:future-sec}
Our algorithmic lower bound shows that online algorithms, operating exclusively based on the information available up to round $t$, cannot surpass $\alpha_{\rm COMP}$. This provides strong evidence for the conjecture that $\alpha_{\rm COMP}$ is the true computational threshold for this model.

At the same time, prior work~\cite{gamarnik2025optimal} made an intriguing observation. For the largest independent set problem in $G(n, \tfrac12)$, they showed that granting the algorithm access to a limited amount of additional information allows it to exceed the computational threshold.\footnote{The resulting algorithm, albeit being online, requires super-polynomial time.} This naturally raises the following question: for the balanced independent set problem in $\ERB(n,p)$, can $\alpha_{\rm COMP}$ be surpassed if the algorithm is permitted a limited number of future queries at each round?

In this section, we show that the answer is yes: online algorithms augmented with a small number of future queries can indeed surpass $\alpha_{\rm COMP}$. To make this precise, we extend the standard definition of an online algorithm (Definition~\ref{def:OnlineAlg}) to allow future queries
\begin{definition}\label{def:cAdmissible}
For $c>0$, let $\mathcal{C}_c$ be the class of online algorithms operating according to Definition~\ref{def:OnlineAlg}, with the following extension. 
At each round $t\in[2n]$, the algorithm may, in addition to observing the edges $(v_t,v)$ for $v\in L_t\cup R_t$, query a (possibly random) set $S_t$ of vertex pairs $\{i,j\}$---which may include vertices not yet revealed---and reveal the status of all edges in $S_t$. The decision is then based on the combined information, where the number of future queries satisfy
\begin{equation}\label{eq:BUDGET}
    \Biggl|\;\bigcup_{t=1}^{2n} S_t \;\cap\; \binom{\A(\ERB(n,p)) }{2}\Biggr|\;\le\; c(\log_b n)^2.
\end{equation}
\end{definition}
Note that $S_t$ may include pairs involving vertices from ${v_{t+1},\dots,v_{2n}}$. For this reason, we refer to the edges in $S_t$ as \emph{future edges}. The total amount of such future information is limited to $O(\log^2 n)$.

Our final main result is as follows.
\begin{theorem}\label{thm:surpass}
    For any $\epsilon>0$ and $p,\gamma\in(0,1)$, there exists $c\coloneqq c(\gamma,\epsilon)$ and an online algorithm  $\A\in \mathcal{C}_c$ that $(k,\delta)$-optimizes the $\gamma$-balanced independent set problem in $\ERB(n,p)$, where
    \[
    k=(1+\epsilon)\alpha_{\rm COMP},\qquad \text{and}\qquad \delta = 1-\exp(-n^{\Theta(1)}).
    \]
\end{theorem}
Our algorithm proceeds in three phases, described informally as follows. 
The first phase runs for $T=o(n)$ rounds and greedily constructs a $\gamma$-balanced independent set $I_T=R_T\cup L_T$ with $R_T\subset R$ and $L_T\subset L$, of size $c'\log_b n/\gamma$ for a suitable constant $c'\coloneqq c'(\gamma,\epsilon)$. The second phase is an \emph{exploration phase}, where—using future queries—it identifies sets $W_L \subset L$ and $W_R \subset R$ not yet inspected, such that there are no edges (i) between $R_T$ and $W_L$, and (ii) between $L_T$ and $W_R$. In the third phase, it performs a brute-force search to identify a $\gamma$-balanced independent set of size $(1+\epsilon-c')\log_b n/\gamma$ inside $\ERB(W_L\cup W_R,p)$, and augments this to $I_T$. 

Our analysis shows that the procedure  articulated above can be implemented in an online fashion; see Section~\ref{sec:Future} for the details. 

Importantly, our algorithm is not fully oblivious to the vertex arrival order. In fact, no online algorithm can produce a $\gamma$-balanced independent set of size $(1+\epsilon)\log_b n$ (whp) while remaining completely oblivious to the arrival order. Suppose, for contradiction, that such an algorithm outputs a $\gamma$-balanced independent set of size $(1+\epsilon)\alpha_{\rm COMP}$. If the first $n$ arrivals all lie on the same side of the bipartition, then by the online constraint (decisions cannot be deferred), the algorithm must select at least $(1+\epsilon)\log_b n$ vertices from that side. This, however, precludes adding vertices from the opposite side, since the expected number of vertices there with no edge to the chosen $(1+\epsilon)\log_b n$ vertices is at most $n^{-\epsilon}$. Thus, some dependence on arrival order is unavoidable. That said, our algorithm's dependence on arrival order is fairly mild: the only step where order matters is in the first greedy phase (Algorithm~\ref{alg:greedy_cont'd}), designed specifically to avoid such pathological situations.
\paragraph{Number of Future Queries}
Our analysis in Section~\ref{sec:Future} shows that it suffices to take
\[
c(\gamma,\epsilon)\coloneqq\frac{1-\gamma}{\gamma}\left((1+\epsilon)^2 - (c')^2\right),
\]
where $c'$ is any arbitrary constant satisfying
\[
c'<\frac{\gamma\bigl(1-(1+\epsilon)(1-\gamma)\bigr)}{(1-\gamma)^2}.
\]
In the balanced case $(\gamma=\tfrac12$), the condition on $c'$ boils down to $c'<1-\epsilon$. In this case, one can choose $c(\frac12,\epsilon)$ such that $c(\frac12,\epsilon)\to 0$ as $\epsilon\to 0$, e.g., by taking $c(\frac12,\epsilon)=6\epsilon$. By contrast, for $\gamma < \tfrac12$, our analysis suggests that it is not possible to choose $c(\gamma,\epsilon)$ such that $c(\gamma,\epsilon)\to 0$ as $\epsilon\to 0$.
At first glance, this may seem surprising.
However, note the following: as $\gamma \to \frac{\epsilon}{1 + \epsilon}$, we have $\alpha_{\rm COMP} \to \alpha_{\rm STAT}$.
In particular, as the statistical-computational gap is ``smaller'' for smaller values of $\gamma$, one expects surpassing the gap to be somewhat ``harder''. Indeed, our analysis reflects this: interpolating between $\gamma = 1/2$ and $\gamma = \frac{\epsilon}{1 + \epsilon}$, we shift from $c(\gamma,\epsilon)=\Theta(\epsilon)$ to $c(\gamma,\epsilon)=\Theta(1/\epsilon)$.\footnote{We remark that the distinction between $\gamma = 1/2$ and $\gamma < 1/2$ is prevalent in the analogous setting of $\gamma$-balanced colorings.
In fact, while the problem has been heavily studied for $\gamma = 1/2$ \cite{chakraborti2023extremal, feige2010balanced, dhawan2025balanced, dhawan2025balancedER}, no results are known for $\gamma$-balanced colorings for $\gamma < 1/2$ (see the discussion in \cite[Section 1.3]{dhawan2025balancedER} on the challenges involved).}

\section{Statistical Threshold: Proof of Theorem~\ref{thm:Max-IS}}\label{section: stat thresh proof}

In this section, we will prove Theorem~\ref{thm:Max-IS}.
Let the vertex set of $G$ be $L \sqcup R$ with $|L|=|R|=n$.
Recall the definition of the random variable $Z_\alpha(\gamma)$.
We define two new variables $Z_\alpha^\eta(\gamma)$ for $\eta \in \set{L, R}$ denoting the number of $\gamma$-balanced independent sets $I$ in $\ERB(n, p)$ such that $\gamma$ proportion of $I$ lies in the partition determined by $\eta$.
Note that
\[Z_\alpha(\gamma) = Z_\alpha^L(\gamma) + Z_\alpha^R(\gamma).\]
Furthermore, it is easy to see by symmetry that $Z_\alpha^L(\gamma) \stackrel{d}{=} Z_\alpha^R(\gamma)$ and so it is enough to consider $Z_\alpha \coloneqq Z_\alpha^L(\gamma)$.

We first prove \ref{stat: FMM} by a simple first moment argument.
Note the following for any $\alpha$:
\begin{equation}\label{eq:FirstMom}
    \mathbb{E}[Z_\alpha] = \binom{n}{\alpha\gamma}\binom{n}{\alpha(1-\gamma)}(1-p)^{\gamma(1-\gamma)\alpha^2}.
\end{equation}
For $\alpha \geq (1+\epsilon)\alpha_{\rm STAT}$, we have
\begin{align*}
    \mathbb{E}[Z_\alpha]&\le n^{\alpha \gamma}n^{\alpha(1-\gamma)}(1-p)^{\gamma(1-\gamma)\alpha^2} \\
    &=\exp\left(\alpha \log n -\gamma(1-\gamma)\alpha^2 \log \frac{1}{1-p}\right) \\
    &\le\exp\left(\alpha\left(\log n - (1+\epsilon)\gamma(1-\gamma)\frac{\log_b n}{\gamma(1-\gamma)}\log\frac{1}{1-p}\right)\right) \\
    &=\exp\left(-\epsilon\alpha \log n\right)=\exp\bigl(-\Omega(\log^2 n)\bigr),
\end{align*}
as desired.

The remainder of this section is dedicated to the proof of \ref{stat: SMM}.
Suppose $\alpha\le (1-\epsilon)\alpha_{\rm STAT}$. It is easy to verify from \eqref{eq:FirstMom}, using the inequality $\binom{n}{k}\ge (n/k)^k$, that 
\[\mathbb{E}[Z_\alpha]=\exp(\Theta(\log^2 n)) = \omega(1).\]
Moreover, as $Z_\alpha \ge Z_{\alpha'}$ for $\alpha\le \alpha'$ it suffices to prove the claim when $\alpha = \alpha_\epsilon\coloneqq (1-\epsilon)\alpha_{\rm STAT}$. We do so by showing $\mathbb{E}[Z_{\alpha_\epsilon}^2]/\mathbb{E}[Z_{\alpha_\epsilon}]^2 = 1+o(1)$.
The result then follows by the Paley-Zygmund inequality.

For notational convenience, set $Z\coloneqq Z_{\alpha_\epsilon}$.
Thus
\[
Z\coloneqq \sum_{\substack{(L',R')\,:\,L'\subset L,\, R'\subset R \\ |L'| = \gamma \alpha_\epsilon,\,|R'|=(1-\gamma)\alpha_\epsilon}} \, I_{(L',R')},
\]
where $I_{(L',R')}$ is the indicator of the event that $(L',R')$ is a $\gamma$-balanced independent set. 
Next,
\begin{align}\label{eq:Z^2}
Z^2\coloneqq \sum_{\substack{(L',R')\,:\,L'\subset L,\, R'\subset R \\ |L'| = \gamma \alpha_\epsilon,\,|R'|=(1-\gamma)\alpha_\epsilon}} \,\,\sum_{\substack{(L'',R'')\,:\,L''\subset L,\, R''\subset R \\ |L''| = \gamma \alpha_\epsilon,\,|R''|=(1-\gamma)\alpha_\epsilon}}\, I_{(L',R')}I_{(L'',R'')}.
\end{align}
In what follows, we use the following parameterization:
\begin{equation}\label{eq:i1-i2}
  i_1\coloneqq |L'\cap L''|\qquad\text{and}\qquad i_2\coloneqq |R'\cap R''|.  
\end{equation}
Clearly, 
\begin{equation}\label{eq:Range-of-is}
    0\le i_1\le \gamma\alpha_\epsilon,\qquad \text{and} \qquad 0\le i_2\le (1-\gamma)\alpha_\epsilon.
\end{equation}
Fix $(L',R')$.
The number of tuples $((L',R'),(L'',R''))$ subject to~\eqref{eq:i1-i2} is
\begin{equation}\label{eq:N-i1i2}
    N(i_1,i_2)\coloneqq \binom{\gamma\alpha_\epsilon}{i_1}\binom{n-\gamma\alpha_\epsilon}{\gamma\alpha_\epsilon-i_1}\binom{(1-\gamma)\alpha_\epsilon}{i_2}\binom{n-(1-\gamma)\alpha_\epsilon}{(1-\gamma)\alpha_\epsilon-i_2}.
\end{equation}
For any such tuple, the quantity $\mathbb{E}[I_{(L',R')}I_{(L'',R'')}]$ depends solely on $(i_1,i_2)$:
\begin{equation}\label{eq:joint-p}
    \mathbb{E}[I_{(L',R')}I_{(L'',R'')}] = (1-p)^{2\gamma(1-\gamma)\alpha_\epsilon^2}(1-p)^{-i_1i_2},
\end{equation}
where the term $(1-p)^{-i_1i_2}$ accounts for the double counted edges.
Combining~\eqref{eq:Z^2}, \eqref{eq:N-i1i2}, and~\eqref{eq:joint-p}, we obtain
\begin{align*}
    \mathbb{E}[Z^2]&=\binom{n}{\gamma\alpha_\epsilon}\binom{n}{(1-\gamma)\alpha_\epsilon}\sum_{i_1 = 0}^{\gamma\alpha_\epsilon} \sum_{i_2= 0}^{(1-\gamma)\alpha_\epsilon}N(i_1,i_2)(1-p)^{2\gamma(1-\gamma)\alpha_\epsilon^2 - i_1i_2}.
\end{align*}
Combining this with~\eqref{eq:FirstMom}, we arrive at
\begin{align}\label{eq:MAIN}
    \frac{\mathbb{E}[Z^2]}{\mathbb{E}[Z]^2} = \sum_{i_1 = 0}^{\gamma\alpha_\epsilon} \sum_{i_2= 0}^{(1-\gamma)\alpha_\epsilon} \frac{\binom{\gamma\alpha_\epsilon}{i_1}\binom{n-\gamma\alpha_\epsilon}{\gamma\alpha_\epsilon-i_1}}{\binom{n}{\gamma\alpha_\epsilon}}\frac{\binom{(1-\gamma)\alpha_\epsilon}{i_2}\binom{n-(1-\gamma)\alpha_\epsilon}{(1-\gamma)\alpha_\epsilon-i_2}}{\binom{n}{(1-\gamma)\alpha_\epsilon}}(1-p)^{-i_1i_2}.
\end{align}
Note the following as a result of the bounds in \eqref{eq:Range-of-is}:
\begin{align*}
    \frac{\binom{n-\gamma\alpha_\epsilon}{\gamma\alpha_\epsilon-i_1}}{\binom{n}{\gamma\alpha_\epsilon}} \leq \frac{\binom{n-i_1}{\gamma\alpha_\epsilon-i_1}}{\binom{n}{\gamma\alpha_\epsilon}} = \frac{(n - i_1)!}{n!}\,\frac{(\gamma\alpha_\epsilon)!}{(\gamma\alpha_\epsilon - i_1)!} \leq \left(\frac{\gamma\alpha_\epsilon}{n}\right)^{i_1},
\end{align*}
and
\[\frac{\binom{n-(1-\gamma)\alpha_\epsilon}{(1-\gamma)\alpha_\epsilon-i_1}}{\binom{n}{\gamma\alpha_\epsilon}} \leq \frac{\binom{n-i_2}{(1-\gamma)\alpha_\epsilon-i_2}}{\binom{n}{(1-\gamma)\alpha_\epsilon}} = \frac{(n - i_2)!}{n!}\,\frac{((1-\gamma)\alpha_\epsilon)!}{((1-\gamma)\alpha_\epsilon - i_2)!} \leq \left(\frac{(1-\gamma)\alpha_\epsilon}{n}\right)^{i_2}.\]
Applying these bounds to \eqref{eq:MAIN}, we have
\begin{align}
    \frac{\mathbb{E}[Z^2]}{\mathbb{E}[Z]^2} &\leq \sum_{i_1 = 0}^{\gamma\alpha_\epsilon} \sum_{i_2= 0}^{(1-\gamma)\alpha_\epsilon} \binom{\gamma\alpha_\epsilon}{i_1}\left(\frac{\gamma\alpha_\epsilon}{n}\right)^{i_1}\binom{(1-\gamma)\alpha_\epsilon}{i_2}\left(\frac{(1-\gamma)\alpha_\epsilon}{n}\right)^{i_2}(1-p)^{-i_1i_2}\nonumber \\
    &\leq \sum_{i_1 = 0}^{\gamma\alpha_\epsilon} \sum_{i_2= 0}^{(1-\gamma)\alpha_\epsilon} \underbrace{\exp\left(-(i_1 + i_2)(\log n - 2\log \alpha_\epsilon) + i_1i_2\log b\right)}_{q(i_1, i_2)}, \label{eq:AUXIL2}
\end{align}
where we use the fact that $\gamma = \Theta(1)$ and $b = 1/(1-p)$.
% Let $q(i_1, i_2)$ denote the terms in the sum above.
To bound $q(i_1, i_2)$, we consider cases.
\begin{enumerate}[label = \textbf{Case \arabic*:}, wide]
    \item $i_1 = i_2 = 0$. Clearly, $q(i_1, i_2) = 1$ in this case.
    \item $i_1 = 0$ and $i_2 \geq 1$ or $i_2 = 0$ and $i_1 \geq 1$. As $\log\alpha_\epsilon = \Theta(\log \log n)$, it follows that $q(i_1, i_2) = \exp\left(-\Omega(\log n)\right)$ in this case.
    \item $i_1, i_2 \geq 1$ and $i_1i_2 \leq \log_bn$.
    In this case, we have
    \[q(i_1, i_2) \leq \exp\left(-2(\log n - 2\log \alpha_\epsilon) + \log n\right) = \exp\left(-\Omega(\log n)\right),\]
    once again.
    \item $i_1, i_2 \geq 1$ and $i_1i_2 \geq \log_bn$.
    Observe that we can further modify $q(i_1, i_2)$ to get
    \begin{align}
        q(i_1,i_2) &= \exp\left(-(i_1 + i_2)(\log n - 2\log \alpha_\epsilon) + i_1i_2\log b\right) \nonumber \\
        &=\exp\left(i_1i_2\log b\left(1 - \left(\frac{1}{i_1}+\frac{1}{i_2}\right)\left(\log_{b}n - 2\log_b\alpha_\epsilon\right)\right)\right), \label{eq:RateFnc1}
    \end{align}
    where we note that $1/i_j$ is well-defined since $i_j \geq 1$ for $j \in \set{1, 2}$.
    Using the bounds on $i_1$ and $i_2$ from~\eqref{eq:Range-of-is}, we control the harmonic mean of $i_1$ and $i_2$ as follows:
    \[\frac{1}{i_1}+\frac{1}{i_2}\ge \frac{1}{\gamma\alpha_\epsilon} + \frac{1}{(1-\gamma)\alpha_\epsilon} = \frac{1}{(1-\epsilon)\log_b n}\ge \frac{1+\epsilon}{\log_b n}.\]
    Consequently,
    \[
    1 - \left(\frac{1}{i_1}+\frac{1}{i_2}\right)\bigl(\log_b n - 2\log_b\alpha_\epsilon\bigr) \le 1-(1+\epsilon)\left(1 - \frac{2\log_b\alpha_\epsilon}{\log_b n}\right) \le -\frac{\epsilon}{2},
    \]
    where we use the fact that $\log \alpha_\epsilon = \Theta(\log \log n)$ and $n$ is sufficiently large in terms of $\epsilon$.
    With this,~\eqref{eq:RateFnc1} is again upper bounded by $\exp\left(-\Omega(\log n)\right)$.
\end{enumerate}

Combining all of the above cases,~\eqref{eq:AUXIL2} becomes:
\begin{align*}
    1\,\leq\, \frac{\mathbb{E}[Z^2]}{\mathbb{E}[Z]^2}&\leq \sum_{i_1 = 0}^{\gamma\alpha_\epsilon} \sum_{i_2= 0}^{(1-\gamma)\alpha_\epsilon}q(i_1,i_2)\\
    &\le 1+\gamma(1-\gamma)\alpha_\epsilon^2 \exp(-\Omega(\log n)) \\
    &=1+\exp\left(-\Omega(\log n)+O(\log \log n)\right)\\
    &=1+\exp\left(-\Omega(\log n)\right),
\end{align*}
as $\gamma=\Theta(1)$ and $\log \alpha_\epsilon = \Theta(\log \log n)$. 
Using the Paley-Zygmund inequality~\cite{alon2016probabilistic},
\[
\mathbb{P}[Z>0]\ge \frac{\mathbb{E}[Z]^2}{\mathbb{E}[Z^2]} = 1-\exp\left(-\Omega(\log n)\right),
\]
completing the proof of \ref{stat: SMM}.

\section{Achievability Result: Proof of Theorem~\ref{theorem: achievability}}\label{sec:achievability}

Recall the statistical threshold $\alpha_{\rm STAT}$ from Theorem~\ref{thm:Max-IS} and \eqref{eq:a-star}. In this section, we give an online algorithm that finds an independent set of size at least $(1-\epsilon)\alpha_{\rm COMP}$ for arbitrary $\epsilon > 0$ with high probability, where $\alpha_{\rm COMP} = (1-\gamma) \alpha_{\rm STAT}$ is as defined in \eqref{eq:alpha-alpg}.
This gives a lower bound on the computational threshold.

Our algorithm will proceed in two stages:
\begin{enumerate}
    \item In \hyperref[alg:greedy]{Stage One}, we greedily find an independent set $I$ satisfying $\max\left\{|I \cap L|,\,|I \cap R|\right\} = (1-\epsilon)\log_bn$.
    \item At this point, by relabeling the partitions if necessary, we may assume $|I \cap L| = (1-\epsilon)\log_b n$.
    During \hyperref[alg:greedy_bal]{Stage Two}, we only add vertices in $R$ to the independent set, stopping once $|I \cap R| = \frac{(1 - \gamma)}{\gamma}(1-\epsilon)\log_b n$.
\end{enumerate}
Before we formally describe the algorithms for each stage, we make the following definition: for a given timestep $t$ and vertex $v$, denote by $N_t(v)$ the set of all its neighbors in $\set{v_1, \ldots, v_{t-1}}$.
Let us now formally describe our greedy algorithm, which constitutes Stage $1$ of our procedure.

\begin{breakablealgorithm}
\caption{Stage One}\label{alg:greedy}

\begin{algorithmic}[1]
    \State Initialize $I_0, L_0, R_0=\varnothing$ and $t = 1$.
    \While{$\max\left\{|I_{t-1} \cap L|,\,|I_{t-1} \cap R|\right\} < (1-\epsilon)\log_bn$}
        \State Sample the random vertex $v_t$.
        \If{$N_t(v_t) \cap I_{t-1} = \varnothing$}
            \State Set $I_t = I_{t-1} \cup \{v_t\}$.
        \Else
            \State $I_t = I_{t-1}$.
        \EndIf
        \If{$v_t \in L$}
            \State $L_t = L_{t-1} \cup \set{v_t}$, $R_t = R_{t-1}$.
        \Else
            \State $L_t = L_{t-1}$, $R_t = R_{t-1} \cup \set{v_t}$.
        \EndIf
        \State Update $t = t + 1$.
    \EndWhile
\end{algorithmic}
\end{breakablealgorithm}

Let $T_f$ be the random variable denoting the number of iterations of the \textsf{while} loop of Algorithm~\ref{alg:greedy}, i.e.,
\[
T_f = \min\bigl\{t\ge 1\,:\,\max\{|I_t\cap L|,|I_t\cap R|=(1-\epsilon)\log_b n\}\bigr\}.
\]
The key result  for Stage One is the following lemma, which shows that $T_f$ is sublinear whp irrespective of the vertex arrival order.

\begin{lemma}\label{lemma: greedy analysis}
    $\mathbb{P}\left[T_f>n^{1-\epsilon/2}\right] = \exp\left(-\Omega\left(n^{\epsilon/2}\right)\right)$.
\end{lemma}

We defer the proof of Lemma~\ref{lemma: greedy analysis} to \S\ref{subsection: greedy}.
Let us now formally describe Stage Two of our online algorithm.
Note that at time $T_f$, the contribution from one of the partitions to the constructed independent set $I_{T_f}$ is $(1-\epsilon)\log_b n$. 
By relabeling the partitions if necessary, we may assume this partition is $L$.
% In general, note that this partition is random, so essentially we are conditioning on it being $L$. 
For the second stage, we employ the following algorithm.

\begin{breakablealgorithm}
\caption{Stage Two}\label{alg:greedy_bal}

\begin{algorithmic}[1]
    % \State Without loss of generality, assume $|I_t \cap L| = (1 - \epsilon)\log_bn$.
    \For{$t = T_f + 1, \ldots, 2n$}
        \State Sample the random vertex $v_t$.
        \If{$v_t \in L$}
            \State $L_t = L_{t-1} \cup \set{v_t}$, $R_t = R_{t-1}$.
        \Else
            \State $L_t = L_{t-1}$, $R_t = R_{t-1} \cup \set{v_t}$.
            \If{$|I_{t-1} \cap R_t| < \frac{(1 - \gamma)}{\gamma}\,(1-\epsilon)\log_bn$ and $N_t(v_t) \cap I_{t-1} = \varnothing$}
                \State Set $I_t = I_{t-1} \cup \{v_t\}$.
            \Else
                \State $I_t = I_{t-1}$.
            \EndIf
        \EndIf
    \EndFor
\end{algorithmic}
\end{breakablealgorithm}

Namely, if $v_t\in L$, we do not add it to the independent set. If $v_t\in R$, we add it to the independent set only if (i) $|I_{t-1}\cap R_t|$ is less than $(1-\gamma)(1-\epsilon)\alpha_{\rm COMP}$, and (ii) $v_t$ has no neighbors in $I_{t-1}$. This truncation ensures the $\gamma$-balancedness constraint.

The key result for Stage Two is the following lemma.

\begin{lemma}\label{lemma: balancing analysis}
    $\mathbb{P}\left[|I_{2n} \cap R| < \frac{(1 - \gamma)}{\gamma}(1-\epsilon)\log_bn\,\bigl\lvert\, T_f \leq n^{1-\epsilon/2}\right] = \exp\left(-\Omega\left(n^{\epsilon/2}\right)\right)$.
\end{lemma}

We defer the proof of Lemma~\ref{lemma: balancing analysis} to \S\ref{subsection: balancing step}.
Before we prove these key lemmas, let us complete the proof of our achievability result.

\begin{proof}[Proof of Theorem~\ref{theorem: achievability}]
    Consider running our two-stage algorithm with input $\ERB(n, p)$.
    Let $I$ be the output and let $T_f$ be the number of iterations of Algorithm~\ref{alg:greedy}.
    Using the inequality $\mathbb{P}[A] \leq \mathbb{P}[B] + \mathbb{P}[A|B^c]$, we have
    \[\mathbb{P}[|I| < (1- \epsilon) \alpha_{\rm COMP}] \,\leq\, \mathbb{P}[T_f > n^{1-\epsilon/2}] + \CProb{|I| < (1- \epsilon) \alpha_{\rm COMP}}{T_f \leq n^{1-\epsilon/2}}.\]
    Note that conditionally on $\{T_f\leq n^{1-\epsilon/2}\}$, the event $|I|<(1-\epsilon)\alpha_{\rm COMP}$ implies the event that $|I\cap R|<\frac{1-\gamma}{\gamma}(1-\epsilon)\log_b n$, where recall that by relabeling the partitions if necessary, we may assume that $I\cap L$ has size $(1-\epsilon)\log_b n$ at time $T_f$. Therefore, combining Lemmas~\ref{lemma: greedy analysis} and \ref{lemma: balancing analysis} completes the proof.
\end{proof}

\subsection{Stage One: Proof of Lemma~\ref{lemma: greedy analysis}}\label{subsection: greedy}

Recall Algorithm~\ref{alg:greedy}. 
To prove Lemma~\ref{lemma: greedy analysis}, we will analyze an alternate process, which does not necessarily produce an independent set, but is easier to analyze and can be coupled with our procedure.

Note that by definition, Algorithm \ref{alg:greedy} terminates at time $t=T_f$. From this time point, we continue the algorithm as follows. Let $\bar{I}_{T_f}\coloneqq I_{T_f}$. At each time step $t\geq T_f+1$, we consider i.i.d.\ random variables $Y_t\sim \mathrm{Ber}(n^{-1+\epsilon})$ for $t \geq T_f + 1$. We update
    \[\bar{I}_{t}= \left\{\begin{array}{cc}
       \bar{I}_{t - 1} \cup \{v_t\}  & \text{if } Y_t = 1; \\
       \bar{I}_{t - 1}  & \text{otherwise.}
    \end{array}\right.\]
As remarked before, $\bar{I}_t$ need not be an independent set for any $t\geq T_f+1$.

Observe that, by definition of $T_f$, for any $t \leq T_f$, the probability with which we add $v_t$ to $I_t$ is
\begin{align*}
    \mathbb{P}[N_t(v_t)\cap I_{t-1}=\varnothing] &= \CProb{\BIN(|I_{t-1}\cap R|,p)=0}{v_t \in L} + \CProb{\BIN(|I_{t-1}\cap L|,p)=0}{v_t \in R} \\ 
    &= (1-p)^{|I_{t-1}\cap R|}\ind{v_t \in L} + (1-p)^{|I_{t-1}\cap L|}\ind{v_t \in R}  \\
    &> (1-p)^{(1-\epsilon)\log_bn} =n^{-1+\epsilon}, \numberthis \label{eq:para_dom}
\end{align*}   
where we recall that $b=1/(1-p)$. Consider an i.i.d.\ sequence of random variables $Z_t \sim \mathrm{Ber}(n^{-1+\epsilon})$ for $t \geq 1$. The observation in \eqref{eq:para_dom} leads us to consider the following simple process:
\begin{itemize}
    \item Initialize $\Tilde{I}_0=\varnothing$.
    \item At each time step $t \geq 1$, update as 
    \[\Tilde{I}_{t}= \left\{\begin{array}{cc}
       \Tilde{I}_{t - 1} \cup \{v_t\}  & \text{if } Z_t = 1; \\
       \Tilde{I}_{t - 1}  & \text{otherwise.}
    \end{array}\right.\]

\end{itemize}
Consider the processes $(\bar{S}_t)_{t \geq 1}$ and $(\Tilde{S}_t)_{t \geq 1}$, where for any $t\geq 1$, we respectively have $\bar{S}_t=|\bar{I}_t|$ and $\Tilde{S}_t=|\Tilde{I}_t|$. As a consequence of \eqref{eq:para_dom}, we note that
\begin{align*}
    q_t\coloneqq \CProb{\bar{S}_{t}=\bar{S}_{t - 1}}{\mathcal{F}_{t - 1}}\leq 1-n^{-1+\epsilon}=\CProb{\Tilde{S}_{t}=\Tilde{S}_{t - 1}}{\mathcal{F}^*_{t - 1}},\numberthis \label{eq:cprob_dom}
\end{align*}
where $\mathcal{F}_t$ (resp.~$\mathcal{F}^*_t$) is the sigma algebra generated by all the information up to (and including) timestep $t$ for the process $(\bar{S}_t)_{t \geq 1}$ (resp.~$(\Tilde{S}_t)_{t \geq 1}$).

Furthermore, as all edges are included independently in $\ERB(n, p)$, we can couple the processes $(\bar{S}_t)_{t \geq 1}$ and $(\Tilde{S}_t)_{t \geq 1}$ on the same probability space using \eqref{eq:cprob_dom} as follows:
\begin{itemize}
    \item Let $\mathcal{Y}_0=0$. For each $t \geq 1$, given $\mathcal{F}_t$, consider independent random variables $\mathcal{X}_t \sim \BER\left(\dfrac{n^{-1+\epsilon}}{1-q_t} \right)$.
    \item For any $t\geq 1$, given $\mathcal{Y}_t$, define $\mathcal{Y}_{t+1}$ as follows. 
    \begin{align*}
        \mathcal{Y}_{t}=\mathcal{Y}_{t - 1} + \mathcal{X}_t\ind{\bar{S}_{t}=\bar{S}_{t - 1}+1}.
    \end{align*}

\end{itemize}
Observe that for any $t\geq 1$, by construction $\mathcal{Y}_{t}\leq \bar{S}_t$, and furthermore, the variables $\mathcal{X}_t$ make sure that marginally, $\mathcal{Y}_t \stackrel{d}{=} \Tilde{S}_t$.

Next, we note that if we define the analogous stopping time 
\begin{align*}
    \bar{T}_f\coloneqq \min\{t\,:\,\max\{|\bar{I}_t\cap L|,|\bar{I}_t\cap R|\}\geq (1-\epsilon)\log_b n\}
\end{align*}
for the set process $(\bar{I}_t)_{t \geq 1}$, then in fact $\bar{T}_f=T_f$, as for any $t\leq T_f$, the construction of the processes $(I_t)_{t \geq 1}$ and $(\bar{I}_t)_{t \geq 1}$ are the same. In particular, to conclude the lemma, it is enough to show that 
\[\mathbb{P}\left[\bar{T}_f>n^{1-\epsilon/2}\right]= \exp\left(-\Omega\left(n^{\epsilon/2}\right)\right).\] 

Since the event $\{|\bar{I}_t| \geq 2(1-\epsilon)\log_b n\}$ implies that either $|\bar{I}_t \cap L| \geq (1-\epsilon)\log_b n$ or $|\bar{I}_t \cap R| \geq (1-\epsilon)\log_b n$, we have 
\begin{align}
    \mathbb{P}\left[\bar{T}_f>n^{1-\epsilon/2}\right]\leq \mathbb{P}\left[\bar{T}_*>n^{1-\epsilon/2}\right], \label{eq:tail_bar}
\end{align}
where we introduce
\begin{align*}
    \bar{T}_*\coloneqq \min\{t\,:\, |\bar{I}_t|=2(1-\epsilon)\log_b n\}=\min\{t \,:\, \bar{S}_t=2(1-\epsilon)\log_b n\}.
\end{align*}
Since $\mathcal{Y}_t\leq \bar{S}_t$ and $\mathcal{Y}_t \stackrel{d}{=} \Tilde{S}_t$, the RHS in \eqref{eq:tail_bar} is at most
\begin{align*}
    \mathbb{P}\left[\Tilde{T}_*>n^{1-\epsilon/2}\right],
\end{align*}
where analogously
\begin{align*}
    \Tilde{T}_*\coloneqq \min\{t \,:\, \Tilde{S}_t=2(1-\epsilon)\log_b n\}.
\end{align*}

Recall the process defining $\Tilde{S}_t=|\Tilde{I}_t|$. 
Observe that irrespective of whether $v_t \in L$ or $v_t \in R$, we always have 
\begin{align*}
    \Tilde{S}_{t}=\Tilde{S}_{t - 1}+Z_t,
\end{align*}
where we recall $(Z_t)_{t \geq 1}$ is an i.i.d.\ sequence with distribution $\BER(n^{-1+\epsilon})$. As a consequence, we can bound
\begin{align*}
    \mathbb{P}\left[\Tilde{T}_*>n^{1-\epsilon/2}\right] \leq \mathbb{P}\left[Z_1+Z_2+\dots+Z_{n^{1-\epsilon/2}}\leq 2(1-\epsilon)\log_b n\right].
\end{align*}
Note that $Z_1+\dots+Z_{n^{1-\epsilon/2}}\sim \BIN(n^{1-\epsilon/2},n^{-1+\epsilon})$ has mean $n^{\epsilon/2}$.
Furthermore, we have
\[n^{\epsilon/2}-2(1-\epsilon)\log_b n>\dfrac{n^{\epsilon/2}}{2},\]
for $n$ sufficiently large. 
A standard Chernoff bound (see, e.g., \cite{alon2016probabilistic}) yields 
\[\mathbb{P}[Z_1+Z_2+\dots+Z_{n^{1-\epsilon/2}}\leq 2(1-\epsilon)\log_b n] = \exp \left(-\Omega\left(n^{\epsilon/2}\right)\right),\]
completing the proof.

\subsection{Stage Two: Proof of Lemma~\ref{lemma: balancing analysis}}\label{subsection: balancing step}

Let us now analyze Algorithm~\ref{alg:greedy_bal} conditioning on the event $T_f \leq n^{1-\epsilon/2}$.
For brevity, we avoid stating this conditioning at every step of the proof.
Once again, we will consider an alternate process $\bar{I}$.
Note that during a single iteration of the \textsf{for} loop of Algorithm \ref{alg:greedy_bal}, we do nothing if $v_t \in L$ or $|I_{t - 1} \cap R| = \frac{(1 - \gamma)}{\gamma}(1-\epsilon)\log_bn$.
Suppose the latter condition first occurs at timestep $T_b$.
From this time point, we continue the algorithm as follows. 
Let $\bar{I}_{T_b - 1}\coloneqq I_{T_b - 1}$. At each time step $t\geq T_b$, we do the following: if $v_t \in L$, we do nothing; if $v_t \in R$, we add $v_t$ to $\bar I_{t - 1}$ if $N_t(v_t) \cap \bar I_{t - 1} = \varnothing$.

Note that as we do nothing for vertices in $L$ and since $|I_{T_f} \cap L| = (1-\epsilon)\log_bn$, we have the following for each vertex $v \in R_f$, where $R_f \coloneqq R \setminus R_{T_f}$:
\[\mathbb{P}[v \in \bar I_{2n}] = (1-p)^{(1-\epsilon)\log_bn} = n^{-1 + \epsilon}.\]
In particular, 
\[\E[|\bar I_{2n} \cap R_f|] = |R_f|n^{-1 + \epsilon} \geq (n - T_f)n^{-1 + \epsilon} \geq n^{\epsilon} - n^{\epsilon/2} \geq n^{\epsilon/2},\]
where we use the assumption that $T_f \leq n^{1 - \epsilon/2}$.
By a simple Chernoff bound (see, e.g., \cite{alon2016probabilistic}), we have
\[\mathbb{P}\left[|\bar I_{2n} \cap R_f| < n^{\epsilon/2}/2\right] = \exp\left(-\Omega\left(n^{\epsilon/2}\right)\right).\]
By the definition of $\bar I_t$, we have
\begin{align*}
    |I_{2n} \cap R| < \frac{(1 - \gamma)}{\gamma}(1-\epsilon)\log_bn &\implies |\bar I_{2n} \cap R| < \frac{(1 - \gamma)}{\gamma}(1-\epsilon)\log_bn \\
    &\implies |\bar I_{2n} \cap R_f| < \frac{(1 - \gamma)}{\gamma}(1-\epsilon)\log_bn.
\end{align*}
From here, we may conclude that
\begin{align*}
    \mathbb{P}\left[|I_{2n} \cap R| < \frac{(1 - \gamma)}{\gamma}(1-\epsilon)\log_bn\right] &\leq \mathbb{P}\left[|\bar I_{2n} \cap R_f| < \frac{(1 - \gamma)}{\gamma}(1-\epsilon)\log_bn\right] \\
    &\leq \mathbb{P}\left[|\bar I_{2n} \cap R_f| < n^{\epsilon/2}/2\right] \\
    &= \exp\left(-\Omega\left(n^{\epsilon/2}\right)\right),
\end{align*}
where we use the fact that $n^{\epsilon/2} \gg \frac{(1 - \gamma)}{\gamma}(1-\epsilon)\log_bn$ as $n \to \infty$.

\section{Impossibility Result: Proof of Theorem~\ref{theorem: impossibility}}\label{sec:impossibility}

In this section we prove that the class of online algorithms falling under Definition~\ref{def:OnlineAlg} fails to produce an independent set of size at least $(1+\epsilon)\alpha_{\rm COMP}$, for any $\epsilon>0$, with suitably high probability. 
This shows that the factor-$1/(1-\gamma)$ statistical-computational gap cannot be bridged by online algorithms. To achieve this result, we study \emph{correlated random graph families}. 
Our proof goes by contradiction. Roughly speaking, under the assumption that an online algorithm $\mathcal{A}$ indeed finds an independent set of size at least $(1+\epsilon)\alpha_{\rm COMP}$, we analyze $\mathcal{A}$ over multiple correlated random graph families to exhibit the existence of a particular structure. The contradiction is then derived by arguing that the probability of such a structure existing must satisfy certain inequalities, which we prove fail to hold.

Note that the online algorithm we consider in Section \ref{sec:achievability} is randomized through the arrival of the vertices $(v_t)_{t \geq 1}$. We first argue that in Theorem \ref{theorem: impossibility}, it is enough to consider algorithms that are \emph{deterministic} instead. To do so, let us formally denote an algorithm $\mathcal{A}$ as 
\begin{align*}
    \mathcal{A}:\{0,1\}^{n^2}\times \Omega \to \{0,1\}^{2n},
\end{align*}
where $\Omega$ is some abstract probability space capturing the randomness coming from the algorithm $\mathcal{A}$, $\{0,1\}^{n^2}$ is the vector indicating which edges are present in the random graph, and the output vector indicates which vertices are present in the constructed independent set by the algorithm $\mathcal{A}$. With a slight abuse of notation, for a bipartite graph $G$ on vertex bipartitions $L$ and $R$, we let $G$ denote both the graph itself as well as the vector $E \in \{0,1\}^{n^2}$ indicating which edges are present.

\begin{lemma}[Reduction to deterministic algorithms]\label{lem:reduc_to_det_algos}
    There exists $\omega^* \in \Omega$ such that the deterministic algorithm $\mathcal{A}:\{0,1\}^{n^2}\to \{0,1\}$ defined as $\mathcal{A}(\cdot)=\mathcal{A}(\cdot,\omega^*)$ satisfies
    \begin{align*}
        \mathbb{P}[|\mathcal{A}(\ERB(n,p),\omega)|\geq (1+\epsilon)\cdot \alpha_{\rm COMP}]\leq\mathbb{P}[|\mathcal{A}(\ERB(n,p))|\geq (1+\epsilon)\cdot \alpha_{\rm COMP}]. 
    \end{align*}
\end{lemma}

\begin{proof}
    Since
    \begin{align*}
        \mathbb{P}[|\mathcal{A}(\ERB(n,p),\omega)|\geq (1+\epsilon)\cdot \alpha_{\rm COMP}]=\E_\omega\big[\CProb{|\mathcal{A}(\ERB(n,p),\omega)|\geq (1+\epsilon)\cdot \alpha_{\rm COMP}}{\omega} \big],
    \end{align*}
    and since by definition
    \begin{align*}
        \mathbb{P}[|\mathcal{A}(\ERB(n,p))|\geq (1+\epsilon)\cdot \alpha_{\rm COMP}]=\CProb{|\mathcal{A}(\ERB(n,p),\omega)|\geq (1+\epsilon)\cdot \alpha_{\rm COMP}}{\omega},
    \end{align*}
    we note that the claimed inequality must hold for some $\omega^* \in \Omega$.
\end{proof}

Recalling Definition~\ref{Def:f_Optimize_ind}, we quickly note that as a consequence of Lemma~\ref{lem:reduc_to_det_algos}, to prove Theorem~\ref{theorem: impossibility}, it is enough to show that no \textit{deterministic} online algorithm $(k,\delta)$-optimizes the balanced independent set problem in $\ERB(n,p)$ for the specified values of $k$ and $\delta$.
Therefore, for the rest of this section, we only consider deterministic online algorithms.

\subsection{Correlated Random Graph Families}\label{section: correlated graphs}

Consider a random graph $G \sim \ERB(n,p)$. 
We will use $G$ as a base graph to define our correlated graph families.

Consider any online algorithm $\mathcal{A}$, and run it on the base graph $G$. 
Let $E_{\mathcal{A}}(T)$ be the set of all the edges of $G$ queried by the algorithm $\mathcal{A}$ up-to and including timestep $T$ for $1 \leq T \leq 2n$, and let $V_{\mathcal{A}}(T)$ be the set of vertices exposed by $\mathcal{A}$ in the first $T$ steps of the algorithm.

For any $m \geq 1$, we define a sequence of random graphs 
\[G_1^{(T)}, \ldots, G_m^{(T)}, \qquad \text{for } 1 \leq T \leq 2n,\]
as follows:
\begin{itemize}
    \item The graph $G_1^{(T)}$ is a copy of $\mathbb G$.
    
    \item For any $2\leq i \leq m$ and any $e \in E_{\mathcal{A}}(T)$, the status of the edge $e$ in $G^{(T)}_i$ is exactly the same as it is in $G$, i.e., it is present in all the graphs $\{G^{(T)}_i\,:\, 2\leq i \leq m\}$ if and only if it is present in $G$, and absent in all the graphs $\{G^{(T)}_i\,:\, 2 \leq i \leq m\}$ otherwise.
    
    \item For any $2\leq i \leq m$ and for any edge $e \not \subseteq V_{\mathcal{A}}(T)$, the status of the edge $e$ is independently decided for each graph $G^{(T)}_i$. In other words, for any $e \not \subseteq V_{\mathcal{A}}(T)$,
    \begin{align*}
        \mathbb{P}[e\text{ is present in }G^{(T)}_i]=\mathbb{P}[\chi^{(T,i)}_{e}=1],
    \end{align*}
    where the collection $\{\chi^{(T,i)}_{e}\,:\,e \not \subseteq V_{\mathcal{A}}(T),\, 1\leq T \leq 2n, \,2\leq i \leq m\}$ is a collection of i.i.d. $\BER(p)$ random variables.
\end{itemize}
We immediately observe the following:

\begin{remark}\label{remark: correlated graphs}
    \mbox{}
    \begin{itemize}
        \item If we consider the entire \emph{graph array} $(G^{(T)}_i)_{1\leq T \leq 2n,1\leq i \leq m}$, then each entry of both the \emph{first column} $(G^{(T)}_1)_{1\leq T\leq 2n}$ and the \emph{last row} $(G^{(2n)}_i)_{1\leq i \leq m}$  is the same as $G$.
        \item For any $1\leq T\leq 2n$ and $1\leq i \leq m$, marginally the distribution of $G^{(T)}_i$ is the same as $G$.
        \item Since $\mathcal{A}$ is deterministic, by construction, the behavior of the first $T$ steps of the algorithm $\mathcal{A}$ is exactly the same on all the graphs $G^{(T)}_1,\ldots,G^{(T)}_m$.
    \end{itemize}
\end{remark}

\subsection{Forbidden tuples of independent sets}

Consider an online algorithm $\mathcal{A}$.
We wish to bound the probability that $\mathcal{A}$ produces an independent set of size at least $(1+\epsilon)\alpha_{\rm COMP}$. 
Fix $\mu = \epsilon^2/2$ and define the stopping time 
\begin{align}
    \tau \coloneqq \min\set{2n,\,\min\{t\,:\,\max\{|I_t\cap L|,|I_t\cap R|\} = (1-\mu)\log_b n\}}. \label{eq:def_tau}
\end{align}
% where $u$ will be suitably chosen later. 
We analyze the algorithm $\mathcal{A}$ on the graphs $G^{(\tau)}_1,\ldots, G^{(\tau)}_m$. Define $\zeta  = L$ (resp. $\zeta' = R$) if $|I_{\tau}\cap L|>|I_{\tau}\cap R|$ and $R$ otherwise (resp. $\zeta' = L$ otherwise). 
Note that $\zeta \in \{L,R\}$ is random and depends on the edges as revealed by the algorithm in the first $\tau$ steps.
We define the \emph{successful event}
\begin{align}\label{eq: sucecesful event}
    \mathcal{S} \coloneqq \bigcap_{1 \leq i \leq m} \left\{|\mathcal{A}(G_i^{\tau})| \geq (1+\epsilon)\alpha_{\rm COMP}\right\} = \bigcap_{1 \leq i \leq m} \mathcal{E}_{i, \tau},
\end{align}
where $\mathcal{E}_{i, T} \coloneqq \left\{|\mathcal{A}(G_i^{T})| \geq (1+\epsilon)\alpha_{\rm COMP}\right\}$.
The key results en route to the proof of Theorem~\ref{theorem: impossibility} provide lower and upper bounds on $\mathbb{P}[\mathcal{S}]$.
We state these results here and prove them in Sections~\ref{subsec: lb proof} and \ref{subsec: ub proof}, respectively.

\begin{proposition}\label{prop: lb}
    Let $\mathcal{E}$ denote the event that $|\mathcal{A}(G)| \geq (1+\epsilon)\alpha_{\rm COMP}$.
    Then $\mathbb{P}[\mathcal{S}] \geq \mathbb{P}[\mathcal{E}]^m$.
\end{proposition}

\begin{proposition}\label{prop: ub}
   Let $m=C\epsilon^{-2}$ for some sufficiently large constant $C>0$. Then, \[\mathbb{P}[\mathcal S] = \exp\left(-\Omega(\log_b^2 n)\right).\]
\end{proposition}

Let us now combine the above propositions to prove our impossibility result.

\begin{proof}[Proof of Theorem~\ref{theorem: impossibility}]
    Suppose there exists an online algorithm $\mathcal{A}$ that $((1+\epsilon)\alpha_{\rm COMP},\, \delta)$-optimizes the $\gamma$-balanced independent set problem in $\ERB(n, p)$.
    Fix $m = C\epsilon^{-2}$ for $C$ sufficiently large.
    By Propositions~\ref{prop: lb} and \ref{prop: ub}, we must have
    \[\delta^m = \exp\left(-\Omega(\log_b^2 n)\right) \implies \delta \leq \exp\left(-\Theta(\epsilon^2\log_b^2 n)\right),\]
    as desired.
\end{proof}

\subsubsection{Lower Bound: Proof of Proposition~\ref{prop: lb}}\label{subsec: lb proof}

In this section, we will prove Proposition~\ref{prop: lb} by a careful application of Jensen's inequality and conditional independence.
Recall the definition of $E_{\mathcal{A}}(T)$.
We have
\begin{align*}
    \mathbb{P}[\tau = T,\, \mathcal{S}] &= \E[\mathbb{P}[\tau = T,\,\mathcal{E}_{1, T}, \ldots, \mathcal{E}_{m, T}]\mid E_{\mathcal{A}}(T)] \\
    &= \E[\ind{\tau = T}\mathbb{P}[\mathcal{E}_{1, T}, \ldots, \mathcal{E}_{m, T}\mid E_{\mathcal{A}}(T)]],
\end{align*}
where we use the fact that $\set{\tau = T}$ is determined by $E_{\mathcal{A}}(T)$ as $\mathcal{A}$ is deterministic.
Note that by the observations made in Remark~\ref{remark: correlated graphs}, the algorithm $\mathcal{A}$ is identical for the first $T$ steps on the graphs $G_1^{(T)}, \ldots, G_m^{(T)}$.
With this in hand, we have
\begin{align*}
    \mathbb{P}[\mathcal{E}_{1, T}, \ldots, \mathcal{E}_{m, T}\mid E_{\mathcal{A}}(T)] = \prod_{i = 1}^m\mathbb{P}[\mathcal{E}_{i, T} \mid E_{\mathcal{A}}(T)] = \mathbb{P}[\mathcal{E}_{1, T} \mid E_{\mathcal{A}}(T)]^m,
\end{align*}
where the last step follows since the random variables $\chi_e^{T, i}$ are i.i.d.
Plugging this in above, we have
\begin{align*}
    \mathbb{P}[\mathcal{S}] &= \sum_{T = 1}^{2n}\mathbb{P}[\tau = T,\, \mathcal{S}] \\
    &= \sum_{T = 1}^{2n}\E[\ind{\tau = T}\mathbb{P}[\mathcal{E}_{1, T}, \ldots, \mathcal{E}_{m, T}\mid E_{\mathcal{A}}(T)]] \\
    &= \sum_{T = 1}^{2n}\E[\ind{\tau = T}\mathbb{P}[\mathcal{E}_{1, T}\mid E_{\mathcal{A}}(T)]^m].
\end{align*}
Observing that $\sum_T\ind{\tau = T} = 1$ and $\ind{\tau=i}\ind{\tau=j}=0$ for $i \neq j$, we can further simplify the above to
\begin{align}
    \mathbb{P}[\mathcal{S}] &= \E\left[\left(\sum_{T = 1}^{2n}\ind{\tau = T}\mathbb{P}[\mathcal{E}_{1, T}\mid E_{\mathcal{A}}(T)]\right)^m\right] \nonumber \\
    &\geq \left(\E\left[\sum_{T = 1}^{2n}\ind{\tau = T}\mathbb{P}[\mathcal{E}_{1, T}\mid E_{\mathcal{A}}(T)]\right]\right)^m, \label{eq: success lb}
\end{align}
where we use Jensen's inequality in the final step.
Once again, since $\sum_T\ind{\tau = T} = 1$, we have
\begin{align*}
    \E\left[\sum_{T = 1}^{2n}\ind{\tau = T}\mathbb{P}[\mathcal{E}_{1, T}\mid E_{\mathcal{A}}(T)]\right] &= \sum_{T = 1}^{2n}\E\left[\ind{\tau = T}\mathbb{P}[\mathcal{E}_{1, T}\mid E_{\mathcal{A}}(T)]\right] \\
    &= \sum_{T = 1}^{2n}\mathbb{P}[\tau = T, \mathcal{E}_{1, T}] = \mathbb{P}[\mathcal{E}].
\end{align*}
Plugging this into \eqref{eq: success lb} completes the proof.

\subsubsection{Upper Bound: Proof of Proposition~\ref{prop: ub}}\label{subsec: ub proof}

In this section, we will prove Proposition \ref{prop: ub}. Recall the the definitions of the successful event $\mathcal{S}$ from \eqref{eq: sucecesful event}, the array $(G_i^{(T)})_{1\leq T \leq 2n, 1\leq i \leq m}$ from Section~\ref{section: correlated graphs}, the stopping time $\tau$ from \eqref{eq:def_tau}, and the random parts $\zeta, \zeta'$ of the vertex bipartition.
Note that on the event $\mathcal{S}$, for any $1\leq i\leq m$, $|I_i|\geq (1+\epsilon)\alpha_{\rm COMP}$, where $I_i=\mathcal{A}(G^{(\tau)}_i)$. 
Recall also that for any $1\leq T \leq 2n$, $E_{\mathcal{A}}(T)$ is the set of edges queried by the algorithm $\mathcal{A}$ in the first $T$ steps, and $V_{\mathcal{A}}(T)$ denotes the set of vertices exposed in the first $T$ steps.

\begin{observation}\label{eq:obs:on_S}
We note that on the successful event $\mathcal{S}$, the following statements hold:
\begin{itemize}
    \item By the construction of $(G^{(T)}_i)_{1\leq T \leq 2n, 1\leq i \leq m}$, the sets $I_i\cap V_{\mathcal{A}}(\tau)$ are identical for $1\leq i \leq m$.
    \item $I_i\cap V_{\mathcal{A}}(\tau)\cap \zeta$ has size $(1-\mu)\log_b n$ for all $1\leq i \leq m$.
    \item $I_i\cap V_{\mathcal{A}}(\tau)\cap \zeta'$ has the same size for all $1\leq i \leq m$. %Let $|I_i\cap V_{\mathcal{A}}(\tau)\cap \zeta'|=\beta \log_b n$, $1\leq i \leq m$, where $\beta<(1-u)$.
\end{itemize}    
\end{observation}

It is convenient at this point to introduce the following set:
\begin{align*}
    A &\coloneqq \set{(a_1, \ldots, a_m) \in \Z^m\,:\, (1+\epsilon)\alpha_{\rm COMP} \leq a_i \leq 2n \text{ for each } i}.
\end{align*}
Next, let us define certain sets which will be convenient for us to decompose the successful event $\mathcal{S}$. Recall $\alpha_{\rm COMP}=\frac{\log_b n}{\gamma}$ and $\mu=\epsilon^2/2$. Additionally, for any set $X$, we denote its power set by $\mathcal{P}(X)$.
\begin{definition}[Forbidden $m$-tuples]\label{def:forbidden_tuples}
    For any $\epsilon>0$, $m \geq 1$, $\mathbf{a}=(a_1,\ldots,a_m)\in A$, $\eta \in \{L,R\}$, and $1\leq T \leq 2n$, we define the set $\chi^{(\eta)}(m,T) \subset \mathcal{P}(L \cup R)^m$
    to be the set of all $m$-tuples $(I_1,\ldots,I_m)$, where each $I_i \subseteq L \cup R$ has the following properties:
    \begin{itemize}
        \item For each $1\leq i\leq m$, $I_i$ is a $\gamma$-balanced independent set in $G_i^{(T)}$.
        \item $|I_i|=a_i$ for each $1\leq i \leq m$.
        \item The sets $I_i\cap V_{\mathcal{A}}(T)$ are identical for $1\leq i \leq m$. 
        \item $|I_i\cap V_{\mathcal{A}}(T)\cap \eta|=(1-\mu)\log_b n$ and $|(I_i\cap V_{\mathcal{A}}(T))\setminus \eta|<(1-\mu)\log_b n$ for all $1\leq i \leq m$.
    \end{itemize}
    Let us further define $X^{\eta}_{m,T}(\mathbf{a})$ to be the size of the set $\chi^{\eta}_{m,T}(\mathbf{a})$.
\end{definition}
Recalling Observation \ref{eq:obs:on_S}, and keeping Definition \ref{def:forbidden_tuples} in mind, we conclude that
\begin{align*}
    \mathcal{S}\subset \bigcup_{\mathbf{a}\in A}\left(\left(\cup_{1\leq T \leq 2n}\{X^{(L)}_{m,T}(\mathbf{a})\geq 1\}\right)\cup\left(\cup_{1\leq T \leq 2n}\{X^{(R)}_{m,T}(\mathbf{a})\geq 1\}\right)\right). \numberthis\label{eq:S_inc}
\end{align*}

Thus to conclude Proposition \ref{prop: ub}, using \eqref{eq:S_inc} and a union bound, it is enough to obtain an upper bound on the sum
\begin{align*}
  \sum_{\mathbf{a}\in A}\left(\sum_{T=1}^{2n}\mathbb{P}[X^{(L)}_{m,T}(\mathbf{a})\geq 1]+\sum_{T=1}^{2n}\mathbb{P}[X^{(R)}_{m,T}(\mathbf{a})\geq 1]\right)
  &=\sum_{\mathbf{a}\in A^+}\sum_{\eta \in \{L,R\}} \left(\sum_{T=1}^{2n}\mathbb{P}[X^{(\eta)}_{m,T}(\mathbf{a})\geq 1]\right).\numberthis \label{eq:split_over_a_ranges}
\end{align*}
To this end, we introduce a new random variable.
For a fixed $\beta \in \N$, $\mathbf{a} \in A$ and $\eta \in \{L,R\}$, let $\chi^{(\eta)}_{m,T}(\mathbf{a},\beta)$ be the set of $m$-tuples $(I_1,\ldots,I_m)$, where for each $1\leq i \leq m$, $I_i$ is a $\gamma$-balanced independent set in $G^{(T)}_i$ with the following properties:
\begin{itemize}
    \item $|I_i|=a_i$ for each $1\leq i \leq m$.
    \item The sets $I_i\cap V_{\mathcal{A}}(T)$ are identical for $1\leq i \leq m$. 
    \item $|I_i\cap V_{\mathcal{A}}(T)\cap \eta|=(1-\mu) \log_b n$ and $|(I_i\cap V_{\mathcal{A}}(T))\setminus \eta|=\beta$ for all $1\leq i \leq m$.
\end{itemize}
We further upper bound \eqref{eq:split_over_a_ranges} by
\begin{align*}
    \sum_{\mathbf{a} \in A}\sum_{\beta = 0}^{(1-\mu)\log_bn}\sum_{\eta \in \{L,R\}} \left(\sum_{T=1}^{2n}\mathbb{P}[X^{(\eta)}_{m,T}(\mathbf{a},\beta)\geq 1]\right),
\end{align*}
where $X^{(\eta)}_{m,T}(\mathbf{a},\beta) = |\chi^{(\eta)}_{m,T}(\mathbf{a},\beta)|$.
The following key lemma provides a bound on $\mathbb{P}[X^{(\eta)}_{m,T}(\mathbf{a},\beta)\geq 1]$.

\begin{lemma}\label{lemma: key prob estimate}
    For any fixed $\mathbf{a} \in A$, $0 \leq \beta \leq (1 - \mu)\log_bn$, $\eta \in \set{L, R}$, and $1 \leq T \leq 2n$, we have
    \[\mathbb{P}[X^{(\eta)}_{m,T}(\mathbf{a},\beta)\geq 1] = \exp\left(-\Omega(\log_b^2n)\right)\]
\end{lemma}

Before we prove the above lemma, we note that it implies our desired result.
Indeed, we have
\begin{align*}
    \sum_{\mathbf{a}\in A}\sum_{\eta \in \{L,R\}} \left(\sum_{T=1}^{2n}\mathbb{P}[X^{(\eta)}_{m,T}(\mathbf{a})\geq 1]\right) &\leq  \sum_{\mathbf{a} \in A}\sum_{\beta = 0}^{(1-\mu)\log_bn}\sum_{\eta \in \{L,R\}} \left(\sum_{T=1}^{2n}\mathbb{P}[X^{(\eta)}_{m,T}(\mathbf{a},\beta)\geq 1]\right) \\
    &= \sum_{\mathbf{a} \in A}\sum_{\beta = 0}^{(1-\mu)\log_bn}\sum_{\eta \in \{L,R\}}\sum_{T=1}^{2n}\exp\left(-\Omega(\log_b^2n)\right) \\
    &\leq (2n)^m\cdot \log_bn \cdot 2 \cdot \exp\left(-\Omega(\log_b^2n)\right) \\
    &= \exp\left(-\Omega(\log_b^2n)\right).
\end{align*}
It follows that \eqref{eq:split_over_a_ranges} is at most $\exp\left(-\Omega(\log^2_bn) \right)$.
Recalling \eqref{eq:S_inc}, this completes the proof of Proposition \ref{prop: ub}.

\begin{proof}[Proof of Lemma~\ref{lemma: key prob estimate}]
    Consider an arbitrary collection of $\gamma$-balanced sets $(I_1, \ldots, I_m) \in \chi^{(\eta)}_{m,T}(\mathbf{a},\beta)$.
    Recall that by definition of $\chi^{(\eta)}_{m,T}(\mathbf{a},\beta)$, we have $I_i \cap V_{\mathcal{A}}(T)$ is identical for each $i$.
    With this in mind, we define the following sets:
    \[\mathcal{I} \coloneqq \set{I \subseteq V_{\mathcal{A}}(T)\,:\, |I \cap \eta| = (1 - \mu)\log_bn,\, |I \setminus \eta| = \beta},\]
    and for each $I \in \mathcal{I}$ and $1 \leq i \leq m$
    \[\mathcal{I}_i(I) \coloneqq \set{I_i \subseteq V(G_i^{(T)}) \setminus V_{\mathcal{A}}(T) \,:\, I \cup I_i \text{ is a $\gamma$-balanced set in } G_i^{(T)} \text{ and } |I \cup I_i| = a_i}.\]
    With these definitions in hand and recalling Remark~\ref{remark: correlated graphs} and Observation~\ref{eq:obs:on_S}, we have
    \begin{align}
        \mathbb{P}[X^{(\eta)}_{m,T}(\mathbf{a},\beta)\geq 1]
        &\leq \sum_{I \in \mathcal{I}}\sum_{\substack{I_i \in \mathcal{I}_i(I) \\ \text{for } 1 \leq i \leq m}}\mathbb{P}[I \cup I_i \text{ is independent in } G_i^{(T)} \text{for each }i] \nonumber \\
        &\leq \sum_{I \in \mathcal{I}}(1-p)^{\beta(1-\mu)\log_bn}\prod_{i = 1}^m\sum_{I_i \in \mathcal{I}_i(I)}p(I, I_i), \label{eq: prob ub for X}
    \end{align}
    where 
    \[p(I, I_i) \coloneqq \CProb{I \cup I_i \text{ is independent in } G_i^{(T)}}{I \text{ is independent in } G_i^{(T)}}.\]
    Consider a fixed $I_i \in \mathcal{I}_i(I)$.
    Let $\theta_i = \gamma\, a_i$ and $\theta_i' = (1 - \gamma)a_i$.
    We have two cases to consider depending on which side $\zeta \in \set{L, R}$ contains $\gamma$ proportion of the balanced independent set $I \cup I_i$.
    \begin{enumerate}[label= \textbf{Case \arabic*}:, wide]
        \item $\zeta = \eta$. 
        Then, we have
        \begin{align*}
            p(I, I_i)
            &= (1 - p)^{(\theta_i - (1-\mu)\log_bn)\theta_i' + (\theta_i' - \beta)(1-\mu)\log_b n} \\
            &= (1-p)^{\theta_i\theta_i' - \beta(1-\mu)\log_bn}
        \end{align*}

        \item $\zeta \neq \eta$. 
        Then, we have
        \begin{align*}
            p(I, I_i)
            &= (1 - p)^{(\theta_i' - (1-\mu)\log_bn)\theta_i + (\theta_i - \beta)(1-\mu)\log_b n} \\
            &= (1-p)^{\theta_i\theta_i' - \beta(1-\mu)\log_bn}
        \end{align*}
    \end{enumerate}
    Combining both cases, by a simple counting argument we have that $\sum_{I_i \in \mathcal{I}_i(I)}p(I, I_i)$ is at most
    \begin{align*}
        % &~\sum_{I_i \in \mathcal{I}_i(I)}p(I, I_i) \\
        &~ (1-p)^{\theta_i\theta_i' - \beta(1-\mu)\log_bn}\left(\binom{n}{\theta_i - (1 - \mu)\log_bn}\binom{n}{\theta_i' - \beta} + \binom{n}{\theta_i' - (1 - \mu)\log_bn}\binom{n}{\theta_i - \beta}\right) \\
        &\leq 2(1-p)^{\theta_i\theta_i' - \beta(1-\mu)\log_bn}\,n^{\theta_i - (1 - \mu)\log_bn + \theta_i' - \beta}.
    \end{align*}
    Plugging the above into \eqref{eq: prob ub for X}, we obtain
    \begin{align*}
        &~\mathbb{P}[X^{(\eta)}_{m,T}(\mathbf{a},\beta)\geq 1] \\
        &\leq 2^m\sum_{I \in \mathcal{I}}(1-p)^{\beta(1-\mu)\log_bn + \sum_{i = 1}^m(\theta_i\theta_i' - \beta(1-\mu)\log_bn)}\,n^{\sum_{i = 1}^m(\theta_i + \theta_i' - (1 - \mu)\log_bn - \beta)} \\
        &= 2^m\binom{n}{(1-\mu)\log_bn}\binom{n}{\beta}(1-p)^{- (m - 1)\beta(1-\mu)\log_bn + \sum_{i = 1}^m\theta_i\theta_i'}\,n^{\sum_{i = 1}^m(\theta_i + \theta_i'  - (1 - \mu)\log_bn - \beta)} \\
        &\leq 2^m(1-p)^{- (m - 1)\beta(1-\mu)\log_bn + \sum_{i = 1}^m\theta_i\theta_i'}\,n^{ - (m-1)((1 - \mu)\log_bn + \beta) + \sum_{i = 1}^m(\theta_i + \theta_i')}.
    \end{align*}
    Recall that $b = 1/(1-p)$.
    The above can be simplified to
    \begin{align*}
        &~ 2^mn^{-(m - 1)(\beta\mu + (1-\mu)\log_bn) - \sum_{i = 1}^m\left(\frac{\theta_i\theta_i'}{\log_b n} - \theta_i - \theta_i'\right)} \\
        &= 2^mn^{-(m - 1)(\beta\mu + (1-\mu)\log_bn) - \sum_{i = 1}^m\left(\left(\frac{\theta_i}{\log_b n} - 1\right)\left(\theta_i' - \log_bn\right) -\log_b n\right)}.
    \end{align*}
    Note that since $\gamma \leq 1/2$, we have
    \[\theta_i,\,\theta_i' \,\geq\, \gamma a_i \,\geq\, (1+\epsilon)\gamma\,\alpha_{\rm COMP} \,\geq\, (1+\epsilon)\log_bn.\]
    With this in hand, we may simplify further to obtain
    \begin{align*}
        \mathbb{P}[X^{(\eta)}_{m,T}(\mathbf{a},\beta)\geq 1] &\leq 2^mn^{-(m - 1)(\beta\mu + (1-\mu)\log_bn) + m(1-\epsilon^2)\log_bn} \\
        &\leq 2^mn^{-(m - 1)(\epsilon^2-\mu)\log_bn + (1-\epsilon^2)\log_bn} \\
        &= 2^mn^{-(m - 1)\epsilon^2\log_bn/2 + (1-\epsilon^2)\log_bn},
    \end{align*}
    where we plug in $\mu = \epsilon^2/2$.
    Recall that $m = C\epsilon^{-2}$.
    In particular, for $C$ sufficiently large, the above is $\exp\left(-\Omega(\log_b^2n)\right)$, as desired.
\end{proof}\section{Future Queries: Proof of Theorem~\ref{thm:surpass}}\label{sec:Future}
Fix an arbitrary constant 
\begin{equation}\label{eq:c'}
c'<\frac{\gamma\bigl(1-(1+\epsilon)(1-\gamma)\bigr)}{(1-\gamma)^2},
\end{equation}
where $\epsilon<\frac{\gamma}{1-\gamma}$. Indeed, otherwise there exists no $\gamma$-balanced independent set of size $(1+\epsilon)\alpha_{\rm COMP}$ per Theorem~\ref{thm:Max-IS}. Note that as $(1+\epsilon)(1-\gamma)\ge 1-\gamma$ and $\gamma\le \frac12$, we have
\begin{equation}\label{eq:trivial-bd-on-c'}
    c'<\frac{\gamma^2}{(1-\gamma)^2}\le 1.
\end{equation}
In what follows, we prove the existence of such an online algorithm with
\[
c(\gamma,\epsilon)=\frac{1-\gamma}{\gamma}\left((1+\epsilon)^2 - (c')^2\right).
\]

\noindent As noted, our algorithm proceeds in three phases. 
\paragraph{Phase I: Greedy Rounds} In phase I, we find a $\gamma$-balanced independent set of size $c'\log_bn /\gamma$ for $c'$ arising in~\eqref{eq:c'}. Proceeding analogously to the proof of Theorem~\ref{theorem: achievability}, we first run Algorithm~\ref{alg:greedy} until time $T_1$ where
\[
T_1 \coloneqq \min\bigl\{t\ge 1:\max\{|I_t\cap L|,|I_t\cap R|\}=c'\log_b n\bigr\}.
\]
Plugging $1-c'$ in place of $\epsilon$ in Lemma~\ref{lemma: greedy analysis}, we deduce
\begin{equation}\label{eq:T_1}
    \mathbb{P}\left[T_1\le n^{1-(1-c')/2}\right] = 1-\exp\left(-\Omega\left(n^{(1-c')/2}\right)\right).
\end{equation}
The bound~\eqref{eq:trivial-bd-on-c'} ensures that~\eqref{eq:T_1} is not vacuous. Note that at time $T_1$, one of the partitions to the independent set $I_{T_1}$ is of size $c'\log_b n$.  By relabeling the partitions if necessary, we again assume this partition is $L$. For the second stage, we employ the following modification of Algorithm~\ref{alg:greedy_bal}. Set $T_2 = n^{c'+\delta}$ where $0<\delta<1-c'$ is arbitrary.

\begin{breakablealgorithm}
\caption{Phase I: Part II}\label{alg:greedy_cont'd}

\begin{algorithmic}[1]
    \For{$t = T_1 + 1, \ldots, T_1+T_2$}
        \State Sample a random vertex $v_t\in R\setminus R_{t-1}$.
        \State $L_t = L_{t-1}$, $R_t = R_{t-1} \cup \set{v_t}$.
        \If{$|I_{t-1} \cap R| < c'\log_b n$ and $N_t(v_t) \cap I_{t-1} = \varnothing$}
            \State Set $I_t = I_{t-1} \cup \{v_t\}$.
        \Else
            \State $I_t = I_{t-1}$.
        \EndIf
    \EndFor
\end{algorithmic}
\end{breakablealgorithm}
\begin{lemma}\label{lemma:T-2}
$\displaystyle\mathbb{P}\left[\left|I_{T_1+T_2}\cap R\right|=c'\log_b n\right]= 1-\exp\left(-\Omega(n^{\delta})\right)$.
\end{lemma}
\begin{proof}
    For $\{v_{T_1+1},\dots,v_{T_1+T_2}\}\subset R$, let $N$ denote the number of $i$ such that $T_1+1\le i\le T_1+T_2$ for which there are no edges between $v_i$ and $I_{T_1}$. We have 
    \[
    N=\sum_{i=T_1+1}^{T_1+T_2} I_i\sim {\rm Bin}(T_2,n^{-c'})
    \]as $I_i$ are iid Bernouli random variables with parameter $n^{-c'}$. Clearly, it suffices to show that
    \[
    \mathbb{P}\bigl[N\ge c'\log_b n\bigr]= 1-\exp\left(-\Omega\left(n^{\delta}\right)\right),
    \]
    as we only add vertices as long as $|I_{t-1}\cap R|<c'\log_b n$. As $\mathbb{E}[N] = n^\delta$, we have
    \[
    \mathbb{P}\bigl[N<c'\log_b n\bigr]\le \mathbb{P}\left[N\le \mathbb{E}[N]/2\right]\le e^{-\mathbb{E}[N]/8} = e^{-n^\delta/8}
    \]
    by a standard Chernoff bound (see, e.g., \cite{alon2016probabilistic}), completing the proof.
\end{proof}
Combining~\eqref{eq:T_1} and Lemma~\ref{lemma:T-2} through a union bound, we have the following with probability $1-\exp\left(-\Omega(n^{\min\{\delta,(1-c')/2\}})\right)$ for $T=T_1+T_2$: (i) $T=o(n)$, (ii) $I_T=I^{(L)}_T\cup I^{(R)}_T$ is a $\gamma$-balanced independent set of size $c'\log_b n/\gamma$ where
\begin{align*}
   |I^{(L)}_T|&\coloneqq|I_T\cap L| = c'\log_b n \\
   |I^{(R)}_T|&\coloneqq|I_T\cap R| = c'\frac{1-\gamma}{\gamma}\log_b n.
\end{align*}
Namely, a $\gamma$-balanced independent set of size $c'\log_b n/\gamma$ is found in $o(n)$ time with high probability. This concludes Phase I of our algorithm. Note that for $t\le T=T_1+T_2$, $S_t = \varnothing$---i.e., no future queries have been made.

In what follows, we condition on the outcome of the greedy rounds.
\paragraph{Phase II: Exploration Phase} 
We condition on the outcome of Phase I and fix a $\theta\in(0,1)$ with
\begin{equation}\label{eq:theta}
(1-\gamma)(1+\epsilon-c')<\theta <1-c'\frac{1-\gamma}{\gamma}.
\end{equation}
Observe that $\theta$ is well defined if 
\[
1-c'\frac{1-\gamma}{\gamma}>(1+\epsilon-c')(1-\gamma)
\]
which is equivalent to~\eqref{eq:c'}. At round $T+1$, we do the following:
\begin{itemize}
    \item Sample a random vertex $v_{T+1}\in [2n]\setminus (L_T\cup R_T)$. If $v_{T+1}\in R$, update $R_{T+1}=R_T\cup \{v_{T+1}\}$ and $L_{T+1}=L_T$; otherwise update $R_{T+1}=R_T$ and $L_{T+1}=L_T\cup \{v_{T+1}\}$.
    \item  Query $S_{T+1}$ involving vertices $\{i,j\}$ with $i\in I^{(R)}_T$ and $j\in L\setminus L_T$ and reveal the edge status of all pairs $(i,j)$.
    \item Using the edges in the step above, identify a $\bar{W}_L\subset L\setminus L_T$ of size $|\bar{W}_L|=n^\theta$ for which $E[\bar{W}_L, I_T^{(R)}] = \varnothing$.
    % such that for any $v\in \bar{W}_L$, there exists no edges between $I_T^{(R)}$ and $v$. 
    Lemma~\ref{lemma:W-defined} below justifies the existence of $\bar{W}_L$ whp.
\end{itemize}
At the end of round $T+1$, we set $I_{T+1}\coloneqq I_T$. Similarly, at round $T+2$, we do the following:
\begin{itemize}
    \item Sample a random vertex $v_{T+2}\in [2n]\setminus (L_{T+1}\cup R_{T+1})$. If $v_{T+2}\in R$, update $R_{T+2}=R_{T+1}\cup \{v_{T+2}\}$ and $L_{T+2}=L_{T+1}$; otherwise update $R_{T+2}=R_{T+1}$ and $L_{T+2}=L_{T+1}\cup \{v_{T+2}\}$.
    \item  Query $S_{T+2}$ involving vertices $\{i,j\}$ with $i\in I^{(L)}_T$ and $j\in R\setminus R_{T+1}$ and reveal the edge status of all pairs $(i,j)$.
    \item Using the edges in the step above, identify a $\bar{W}_R\subset R\setminus R_{T+1}$ of size $|\bar{W}_R|=n^\theta$ for which $E[I_T^{(L)}, \bar{W}_R] = \varnothing$.
    % such that for any $v\in \bar{W}_R$, there exists no edges between $I_T^{(L)}$ and $v$. 
    Lemma~\ref{lemma:W-defined} justifies the existence of $\bar{W}_R$ whp.
\end{itemize}
We then set $I_{T+2}\coloneqq I_{T+1}$.
\begin{lemma}\label{lemma:W-defined}
    For $\theta$ satisfying~\eqref{eq:theta}, such sets $\bar{W}_L$ and $\bar{W}_R$ indeed exist with probability at least $1-\exp\left(-n^{\Theta(1)}\right)$.
\end{lemma}
\begin{proof}
    Let
    \[
    N_L = \sum_{v\in L\setminus L_T}\ind{v \text{ has no neighbors in } I_T^{(R)}}
    \]
    Note that for $v\in L\setminus L_T$
    \[
    \ind{v \text{ has no neighbors in } I_T^{(R)}} \stackrel{\text{iid}}{\sim} {\rm Ber}\left(n^{1-c'\frac{1-\gamma}{\gamma}}\right).
    \]
    Since we deterministically have $|L_T|,|R_T|\le T$, and $T=o(n)$ by conditioning,
    $N_L$ stochastically dominates ${\rm Bin}(n-T,n^{-c'(1-\gamma)/\gamma})$, which has mean $n^{1-c'(1-\gamma)/\gamma}(1+o(1))$ (see, e.g.,~\cite[Chapter~4]{roch2024modern}). Proceeding identically to the proof of Lemma~\ref{lemma:T-2},
    \begin{equation}\label{eq:NL}
            \mathbb{P}\left[N_L\ge n^\theta\right] \ge \mathbb{P}\left[{\rm Bin}\left(n-T,n^{-c'\frac{1-\gamma}{\gamma}}\right)\ge n^\theta\right]\ge 1-\exp\left(-n^{\Theta(1)}\right)
        \end{equation}
    as $\theta$ satisfies~\eqref{eq:theta}. Similarly, define
    \[
    N_R = \sum_{v\in R\setminus R_{T+1}}\ind{v \text{ has no neighbors in } I_T^{(L)}}.
    \]
    Clearly, for $v\in R\setminus R_{T+1}$
    \[
    \ind{v \text{ has no neighbors in } I_T^{(L)}} \stackrel{\text{iid}}{\sim} {\rm Ber}\left(n^{1-c'}\right).
    \]
As $|R\setminus R_{T+1}|  = |R|-|R_{T+1}|\ge n-T-1$, and therefore $N_R$ stochastically dominates ${\rm Bin}(n-T-1,n^{-c'})$ which has mean $n^{1-c'}(1+o(1))$. Proceeding identically as above, we obtain that
\begin{equation}\label{eq:NR}
    \mathbb{P}\left[N_R\ge n^\theta\right]\ge 1-\exp\left(-n^{\Theta(1)}\right).
\end{equation}
Noting that such $\bar{W}_L$ and $\bar{W}_R$ with sizes $n^\theta$ exist if and only if $N_L,N_R\ge n^\theta$, we conclude the proof by combining~\eqref{eq:NL} and~\eqref{eq:NR} through a union bound.
\end{proof}
\paragraph{Phase III: Brute-Force Step} At round $t=T+3$, we:
\begin{itemize}
    \item Sample a vertex $v_{T+3}\in [2n]\setminus (L_{T+2}\cup R_{T+2})$. If $v_{T+3}\in R$, update $R_{T+3}=R_{T+2}\cup \{v_{T+3}\}$ and $L_{T+3}=L_{T+2}$; otherwise update $R_{T+3}=R_{T+2}$ and $L_{T+3}=L_{T+2}\cup \{v_{T+3}\}$.
    \item {\bf (Pre-Processing)} Set $W_L\coloneqq\bar{W}_L\setminus \{v_{T+1},v_{T+2},v_{T+3}\}$ and $W_R\coloneqq\bar{W}_R\setminus \{v_{T+1},v_{T+2},v_{T+3}\}$. Note that $|W_L-W_R|\le 3$. Removing at most three vertices from the larger set if necessary, we may assume that $|W_L|=|W_R|=n^\theta+O(1)$.
    \item {\bf (Brute-Force Search)} Query $S_{T+3}$ involving vertices $\{i,j\}$ where $i\in W_L$ and $j\in W_R$ and reveal the edge status of all pairs $(i,j)$. Using this information, identify a $\gamma$-balanced independent set $\mathcal{J}=J_L\cup J_R$ in $G[W_L\cup W_R]$ $\ERB(n^\theta+O(1),p)$) of size $(1+\epsilon-c')\log_b n/\gamma$, where
    \begin{align*}
    |J_L|&=|\mathcal{J}\cap L| = (1+\epsilon-c')\log_b n, \\
    |J_R|&=|\mathcal{J}\cap R| = (1+\epsilon-c')\frac{1-\gamma}{\gamma}\log_b n.
    \end{align*}
    Lemma~\ref{lemma:brute} justifies the existence of $\mathcal{J}$ whp.
    We note that this step has quasi-polynomial running time $n^{O(\log n)}$ as mentioned in Section~\ref{sec:future-sec}.
\end{itemize}
Set $I_{T+3}\coloneqq I_T$.
\begin{lemma}\label{lemma:brute}
   The bipartite random graph $G[W_L\cup W_R] \sim \ERB(n^\theta + O(1),p)$ with vertex set $W_L\cup W_R$
   % , where $W_L\cap W_R=\varnothing$ and $|W_L|=|W_R|=n^\theta + O(1)$ 
   contains a $\gamma$-balanced independent set of size $(1+\epsilon-c')\log_b n/\gamma$.
\end{lemma}
\begin{proof}
Applying Theorem~\ref{thm:Max-IS} to $\ERB(n^\theta + O(1),p)$, it suffices to verify that
\[
\frac{\theta}{\gamma(1-\gamma)}\log_b n>(1+\epsilon-c')\frac{\log_b n}{\gamma} \Leftrightarrow \theta>(1-\gamma)(1+\epsilon-c'),
\]
which holds due to~\eqref{eq:theta}.
\end{proof}
\paragraph{Final Phases} For rounds $T+4\le t\le 2n$:
\begin{itemize}
    \item We randomly sample a vertex $v_t\in[2n]\setminus (L_{t-1}\cup R_{t-1})$. If $v_t\in R$, update $R_t = R_{t-1}\cup \{v_t\}$ and $L_t = L_{t-1}$. Otherwise, set $R_t = R_{t-1}$ and $L_t = L_{t-1}\cup \{v_t\}$. 
    \item If $v_t\in \mathcal{J}$, set $I_{t+1}=I_t\cup \{v_t\}$. Else, set $I_{t+1}=I_t$.
\end{itemize}
The algorithm above is online and constructs a $\gamma$-balanced independent set of size $(1+\epsilon)\alpha_{\rm COMP}$. We finally verify the constraint~\eqref{eq:BUDGET} on the number of queries arising in Definition~\ref{def:cAdmissible}.
\paragraph{Number of Future Queries}
Note that $S_t = \varnothing$ except $t\in\{T+1,T+2,T+3\}$. We have:
\begin{align*}
    \left|\bigcup_{t=1}^{2n} S_t \cap \binom{\A(\ERB(n,p))}{2}\right|&\le \sum_{t=T+1}^{T+3}\left|S_t \cap \binom{\A(\ERB(n,p))}{2}\right| \\
    &=|I_T^{(L)}|\cdot |J_R| + |I_T^{(R)}|\cdot |J_L| + |J_L|\cdot |J_R|\\
    &=\frac{2c'(1+\epsilon-c')(1-\gamma)}{\gamma}(\log_b n)^2 + \frac{(1+\epsilon-c')^2(1-\gamma)}{\gamma}(\log_b n)^2\\
    &=\frac{1-\gamma}{\gamma}\left((1+\epsilon)^2 - (c')^2\right)(\log_b n)^2\\
    &=c(\gamma,\epsilon)(\log_b n)^2.
\end{align*}
This completes the proof of Theorem~\ref{thm:surpass}.

\printbibliography

@book{roch2024modern,
  title={Modern discrete probability: An essential toolkit},
  author={Roch, Sebastien},
  year={2024},
  publisher={Cambridge University Press}
}

@article{vafa2025symmetric,
  title={Symmetric Perceptrons, Number Partitioning and Lattices},
  author={Vafa, Neekon and Vaikuntanathan, Vinod},
  journal={arXiv preprint arXiv:2501.16517},
  year={2025}
}

@article{li2024some,
  title={Some easy optimization problems have the overlap-gap property},
  author={Li, Shuangping and Schramm, Tselil},
  journal={arXiv preprint arXiv:2411.01836},
  year={2024}
}

@article{du2025algorithmic,
  title={The algorithmic phase transition of random graph alignment problem},
  author={Du, Hang and Gong, Shuyang and Huang, Rundong},
  journal={Probability Theory and Related Fields},
  pages={1--56},
  year={2025},
  publisher={Springer}
}

@article{huang2025strong,
  title={Strong Low Degree Hardness for Stable Local Optima in Spin Glasses},
  author={Huang, Brice and Sellke, Mark},
  journal={arXiv preprint arXiv:2501.06427},
  year={2025}
}

@article{gamarnik2025turing,
  title={Turing in the shadows of Nobel and Abel: an algorithmic story behind two recent prizes},
  author={Gamarnik, David},
  journal={arXiv preprint arXiv:2501.15312},
  year={2025}
}

@inproceedings{matula1970complete,
  title={On the complete subgraphs of a random graphProc},
  author={Matula, David W},
  booktitle={Second Chapel Hill Conf. on Combinatorial Mathematics and its Applications (Univ. North Carolina, Chapel Hill, NC, 1970)},
year={1970},
  pages={356--369}
}

@article{gamarnik2022disordered,
  title={Disordered systems insights on computational hardness},
  author={Gamarnik, David and Moore, Cristopher and Zdeborov{\'a}, Lenka},
  journal={Journal of Statistical Mechanics: Theory and Experiment},
  volume={2022},
  number={11},
  pages={114015},
  year={2022},
  publisher={IOP Publishing}
}

@book{matula1976largest,
  title={The largest clique size in a random graph},
  author={Matula, David W},
  year={1976},
  publisher={Department of Computer Science, Southern Methodist University Dallas, TX}
}

@article{lauer2007large,
  title={Large independent sets in regular graphs of large girth},
  author={Lauer, Joseph and Wormald, Nicholas},
  journal={Journal of Combinatorial Theory, Series B},
  volume={97},
  number={6},
  pages={999--1009},
  year={2007},
  publisher={Elsevier}
}

@inproceedings{bollobas1976cliques,
  title={Cliques in random graphs},
  author={Bollob{\'a}s, B{\'e}la and Erd{\"o}s, Paul},
  booktitle={Mathematical Proceedings of the Cambridge Philosophical Society},
  volume={80},
  number={3},
  pages={419--427},
  year={1976},
  organization={Cambridge University Press}
}

@article{li2024discrepancy,
  title={Discrepancy algorithms for the binary perceptron},
  author={Li, Shuangping and Schramm, Tselil and Zhou, Kangjie},
  journal={arXiv preprint arXiv:2408.00796},
  year={2024}
}

@article{gamarnik2023algorithmic,
  title={Algorithmic obstructions in the random number partitioning problem},
  author={Gamarnik, David and K{\i}z{\i}lda{\u{g}}, Eren C},
  journal={The Annals of Applied Probability},
  volume={33},
  number={6B},
  pages={5497--5563},
  year={2023},
  publisher={Institute of Mathematical Statistics}
}

@article{wu2018statistical,
  title={Statistical problems with planted structures: Information-theoretical and computational limits},
  author={Wu, Yihong and Xu, Jiaming},
  journal={arXiv preprint arXiv:1806.00118},
  year={2018}
}

@book{alon2016probabilistic,
  title={The probabilistic method},
  author={Alon, Noga and Spencer, Joel H},
  year={2016},
  publisher={John Wiley \& Sons}
}

@article{gamarnik2023shattering,
  title={Shattering in the Ising Pure $ p $-Spin Model},
  author={Gamarnik, David and Jagannath, Aukosh and K{\i}z{\i}lda{\u{g}}, Eren C},
  journal={arXiv preprint arXiv:2307.07461},
  year={2023}
}

@article{huang2021tight,
  title={Tight Lipschitz Hardness for Optimizing Mean Field Spin Glasses},
  author={Huang, Brice and Sellke, Mark},
  journal={arXiv preprint arXiv:2110.07847},
  year={2021}
}

@article{rahman2017local,
  title={Local algorithms for independent sets are half-optimal},
  author={Rahman, Mustazee and Virag, Balint},
  journal={The Annals of Probability},
  volume={45},
  number={3},
  pages={1543--1577},
  year={2017},
  publisher={Institute of Mathematical Statistics}
}

@article{wein2020optimal,
  title={Optimal low-degree hardness of maximum independent set},
  author={Wein, Alexander S},
  journal={Mathematical Statistics and Learning},
  volume={4},
  number={3},
  pages={221--251},
  year={2022}
}

@inproceedings{ajtai1996generating,
  title={Generating hard instances of lattice problems},
  author={Ajtai, Mikl{\'o}s},
  booktitle={Proceedings of the twenty-eighth annual ACM symposium on Theory of computing},
  pages={99--108},
  year={1996}
}

@article{GK-SK-AAP,
author = {David Gamarnik and Eren C. Kızıldağ},
title = {{Computing the partition function of the Sherrington–Kirkpatrick model is hard on average}},
volume = {31},
journal = {The Annals of Applied Probability},
number = {3},
publisher = {Institute of Mathematical Statistics},
pages = {1474 -- 1504},
year = {2021},
doi = {10.1214/20-AAP1625},
URL = {https://doi.org/10.1214/20-AAP1625}
}

@article{rakhlin2010online,
  title={Online learning: Random averages, combinatorial parameters, and learnability},
  author={Rakhlin, Alexander and Sridharan, Karthik and Tewari, Ambuj},
  journal={Advances in neural information processing systems},
  volume={23},
  year={2010}
}

@inproceedings{rakhlin2011online-a,
  title={Online learning: Beyond regret},
  author={Rakhlin, Alexander and Sridharan, Karthik and Tewari, Ambuj},
  booktitle={Proceedings of the 24th Annual Conference on Learning Theory},
  pages={559--594},
  year={2011},
  organization={JMLR Workshop and Conference Proceedings}
}

@inproceedings{rakhlin2013online,
  title={Online learning with predictable sequences},
  author={Rakhlin, Alexander and Sridharan, Karthik},
  booktitle={Conference on Learning Theory},
  pages={993--1019},
  year={2013},
  organization={PMLR}
}

@article{hazan2016introduction,
  title={Introduction to online convex optimization},
  author={Hazan, Elad and others},
  journal={Foundations and Trends{\textregistered} in Optimization},
  volume={2},
  number={3-4},
  pages={157--325},
  year={2016},
  publisher={Now Publishers, Inc.}
}

@article{rakhlin2011online-b,
  title={Online learning: Stochastic, constrained, and smoothed adversaries},
  author={Rakhlin, Alexander and Sridharan, Karthik and Tewari, Ambuj},
  journal={Advances in neural information processing systems},
  volume={24},
  year={2011}
}

@inproceedings{achlioptas2008algorithmic,
  title={Algorithmic barriers from phase transitions},
  author={Achlioptas, Dimitris and Coja-Oghlan, Amin},
  booktitle={2008 49th Annual IEEE Symposium on Foundations of Computer Science},
  pages={793--802},
  year={2008},
  organization={IEEE}
}

@article{karp1976probabilistic,
  title={The probabilistic analysis of some combinatorial search algorithms.},
  author={Karp, Richard M.},
  year={1976}
}

@article{frieze1990independence,
  title={On the independence number of random graphs},
  author={Frieze, Alan M},
  journal={Discrete Mathematics},
  volume={81},
  number={2},
  pages={171--175},
  year={1990},
  publisher={Elsevier}
}

@article{frieze1992independence,
  title={On the independence and chromatic numbers of random regular graphs},
  author={Frieze, Alan M and Luczak, T},
  journal={Journal of Combinatorial Theory, Series B},
  volume={54},
  number={1},
  pages={123--132},
  year={1992},
  publisher={Elsevier}
}

@article{bandeira2018notes,
  title={Notes on computational-to-statistical gaps: predictions using statistical physics},
  author={Bandeira, Afonso S and Perry, Amelia and Wein, Alexander S},
  journal={arXiv preprint arXiv:1803.11132},
  year={2018}
}

@article{boix2021average,
  title={The Average-Case Complexity of Counting Cliques in Erd{\"o}s-R\'{e}nyi Hypergraphs},
  author={Boix-Adser{\`a}, Enric and Brennan, Matthew and Bresler, Guy},
  journal={SIAM Journal on Computing},
  number={0},
  pages={FOCS19--39},
  year={2021},
  publisher={SIAM}
}

@article{bhamidi2025finding,
  title={Finding a dense submatrix of a random matrix. Sharp bounds for online algorithms},
  author={Bhamidi, Shankar and Gamarnik, David and Gong, Shuyang},
  journal={arXiv preprint arXiv:2507.19259},
  year={2025}
}

@article{mallarapu2025strong,
  title={Strong Low Degree Hardness for the Number Partitioning Problem},
  author={Mallarapu, Rushil and Sellke, Mark},
  journal={arXiv preprint arXiv:2505.20607},
  year={2025}
}

@article{ding2023low,
  title={Low-Degree Hardness of Detection for Correlated Erd$\backslash$H $\{$o$\}$ sR$\backslash$'enyi Graphs},
  author={Ding, Jian and Du, Hang and Li, Zhangsong},
  journal={arXiv preprint arXiv:2311.15931},
  year={2023}
}

@article{hatami2014limits,
  title={Limits of locally--globally convergent graph sequences},
  author={Hatami, Hamed and Lov{\'a}sz, L{\'a}szl{\'o} and Szegedy, Bal{\'a}zs},
  journal={Geometric and Functional Analysis},
  volume={24},
  number={1},
  pages={269--296},
  year={2014},
  publisher={Springer}
}

@article{wein2025computational,
  title={Computational Complexity of Statistics: New Insights from Low-Degree Polynomials},
  author={Wein, Alexander S},
  journal={arXiv preprint arXiv:2506.10748},
  year={2025}
}

@article{kizildaug2023sharp,
      title={{Sharp Thresholds for the Overlap Gap Property: Ising $p$-Spin Glass and Random $k$-SAT}}, 
  author={K{\i}z{\i}lda{\u{g}}, Eren C},
  journal={arXiv preprint arXiv:2309.09913},
  year={2025}
}

@article{huang2023algorithmic,
  title={Algorithmic threshold for multi-species spherical spin glasses},
  author={Huang, Brice and Sellke, Mark},
  journal={arXiv preprint arXiv:2303.12172},
  year={2023}
}

@article{sellke2025tight,
  title={Tight Low Degree Hardness for Optimizing Pure Spherical Spin Glasses},
  author={Sellke, Mark},
  journal={arXiv preprint arXiv:2504.04632},
  year={2025}
}

@inproceedings{gamarnik2022algorithms,
  title={Algorithms and barriers in the symmetric binary perceptron model},
  author={Gamarnik, David and K{\i}z{\i}lda{\u{g}}, Eren C and Perkins, Will and Xu, Changji},
  booktitle={2022 IEEE 63rd Annual Symposium on Foundations of Computer Science (FOCS)},
  pages={576--587},
  year={2022},
  organization={IEEE}
}

@inproceedings{gamarnik2023geometric,
  title={Geometric barriers for stable and online algorithms for discrepancy minimization},
  author={Gamarnik, David and Kizilda{\u{g}}, Eren C and Perkins, Will and Xu, Changji},
  booktitle={The Thirty Sixth Annual Conference on Learning Theory},
  pages={3231--3263},
  year={2023},
  organization={PMLR}
}

@inproceedings{bayati2010combinatorial,
  title={Combinatorial approach to the interpolation method and scaling limits in sparse random graphs},
  author={Bayati, Mohsen and Gamarnik, David and Tetali, Prasad},
  booktitle={Proceedings of the forty-second ACM symposium on Theory of computing},
  pages={105--114},
  year={2010}
}

@article{jerrum1992large,
  title={Large cliques elude the Metropolis process},
  author={Jerrum, Mark},
  journal={Random Structures \& Algorithms},
  volume={3},
  number={4},
  pages={347--359},
  year={1992},
  publisher={Wiley Online Library}
}

@article{gamarnik2017,
author = "Gamarnik, David and Sudan, Madhu",
doi = "10.1214/16-AOP1114",
fjournal = "Annals of Probability",
journal = "Ann. Probab.",
month = "07",
number = "4",
pages = "2353--2376",
publisher = "The Institute of Mathematical Statistics",
title = "Limits of local algorithms over sparse random graphs",
url = "https://doi.org/10.1214/16-AOP1114",
volume = "45",
year = "2017"
}

@article{mezard2005clustering,
  title={Clustering of solutions in the random satisfiability problem},
  author={M{\'e}zard, Marc and Mora, Thierry and Zecchina, Riccardo},
  journal={Physical Review Letters},
  volume={94},
  number={19},
  pages={197205},
  year={2005},
  publisher={APS}
}

@inproceedings{achlioptas2006solution,
  title={On the solution-space geometry of random constraint satisfaction problems},
  author={Achlioptas, Dimitris and Ricci-Tersenghi, Federico},
  booktitle={Proceedings of the thirty-eighth annual ACM symposium on Theory of computing},
  pages={130--139},
  year={2006}
}

@article{gamarnikjagannath2021overlap,
  title={The overlap gap property and approximate message passing algorithms for $ p $-spin models},
  author={Gamarnik, David and Jagannath, Aukosh},
  journal={The Annals of Probability},
  volume={49},
  number={1},
  pages={180--205},
  year={2021},
  publisher={Institute of Mathematical Statistics}
}

@article{gamarnik2021circuit,
  title={Circuit Lower Bounds for the p-Spin Optimization Problem},
  author={Gamarnik, David and Jagannath, Aukosh and Wein, Alexander S},
  journal={arXiv preprint arXiv:2109.01342},
  year={2021}
}

@article{gamarnik2018finding,
  title={Finding a large submatrix of a Gaussian random matrix},
  author={Gamarnik, David and Li, Quan},
  journal={The Annals of Statistics},
  volume={46},
  number={6A},
  pages={2511--2561},
  year={2018},
  publisher={Institute of Mathematical Statistics}
}

@article{gamarnik2017performance,
  title={Performance of sequential local algorithms for the random NAE-K-SAT problem},
  author={Gamarnik, David and Sudan, Madhu},
  journal={SIAM Journal on Computing},
  volume={46},
  number={2},
  pages={590--619},
  year={2017},
  publisher={SIAM}
}

@article{chen2019suboptimality,
  title={Suboptimality of local algorithms for a class of max-cut problems},
  author={Chen, Wei-Kuo and Gamarnik, David and Panchenko, Dmitry and Rahman, Mustazee},
  journal={The Annals of Probability},
  volume={47},
  number={3},
  pages={1587--1618},
  year={2019},
  publisher={Institute of Mathematical Statistics}
}

@inproceedings{gamarnik2020low,
  title={Low-degree hardness of random optimization problems},
  author={Gamarnik, David and Jagannath, Aukosh and Wein, Alexander S},
  booktitle={2020 IEEE 61st Annual Symposium on Foundations of Computer Science (FOCS)},
  pages={131--140},
  year={2020},
  organization={IEEE}
}

@article{bresler2021algorithmic,
  title={The Algorithmic Phase Transition of Random $ k $-SAT for Low Degree Polynomials},
  author={Bresler, Guy and Huang, Brice},
  journal={arXiv preprint arXiv:2106.02129},
  year={2021}
}

@inproceedings{gamarnik2014limits,
  title={Limits of local algorithms over sparse random graphs},
  author={Gamarnik, David and Sudan, Madhu},
  booktitle={Proceedings of the 5th conference on Innovations in theoretical computer science},
  pages={369--376},
  year={2014}
}

@article{gamarnik2021overlap,
  title={The overlap gap property: A topological barrier to optimizing over random structures},
  author={Gamarnik, David},
  journal={Proceedings of the National Academy of Sciences},
  volume={118},
  number={41},
  year={2021},
  publisher={National Acad Sciences}
}

@book{gareyjohnson,
title={Computers and Intractability; A Guide to the Theory of NP-Completeness},
author={Garey, Michael R. and Johnson, David S.},
year={1990},
publisher={W. H. Freeman \& Co.},
address={USA}
}

@inproceedings{feige2002relations,
  title={Relations between average case complexity and approximation complexity},
  author={Feige, Uriel},
  booktitle={Proceedings of the thiry-fourth annual ACM symposium on Theory of computing},
  pages={534--543},
  year={2002}
}

@article{mezard1987mean,
  title={Mean-field theory of randomly frustrated systems with finite connectivity},
  author={M{\'e}zard, Marc and Parisi, Giorgio},
  journal={Europhysics Letters},
  volume={3},
  number={10},
  pages={1067},
  year={1987},
  publisher={IOP Publishing}
}

@inproceedings{AIM2024,
  title={Problem List},
  booktitle={Problem List of the AIM Workshop Low-degree polynomial methods in average-case complexity (organizers Sam Hopkins, Tselil Schramm, and Alex Wein)},
  year={2024}
}

@inproceedings{perkins2024hardness,
  title={On the hardness of finding balanced independent sets in random bipartite graphs},
  author={Perkins, Will and Wang, Yuzhou},
  booktitle={Proceedings of the 2024 Annual ACM-SIAM Symposium on Discrete Algorithms (SODA)},
  pages={2376--2397},
  year={2024},
  organization={SIAM}
}

@article{dhawan2024low,
  title={The low-degree hardness of finding large independent sets in sparse random hypergraphs},
  author={Dhawan, Abhishek and Wang, Yuzhou},
  journal={arXiv preprint arXiv:2404.03842},
  year={2024}
}

@inproceedings {F2014,
    AUTHOR = {Frieze, Alan},
     TITLE = {Random structures and algorithms},
 BOOKTITLE = {Proceedings of the {I}nternational {C}ongress of
              {M}athematicians---{S}eoul 2014. {V}ol. 1},
     PAGES = {311--340},
 PUBLISHER = {Kyung Moon Sa, Seoul},
      YEAR = {2014},
}

@article{gamarnik2025optimal,
  title={Optimal Hardness of Online Algorithms for Large Independent Sets},
  author={Gamarnik, David and K{\i}z{\i}lda{\u{g}}, Eren C and Warnke, Lutz},
  journal={arXiv preprint arXiv:2504.11450},
  year={2025}
}

@inproceedings{khot2001improved,
  title={Improved inapproximability results for maxclique, chromatic number and approximate graph coloring},
  author={Khot, Subhash},
  booktitle={Proceedings 42nd IEEE Symposium on Foundations of Computer Science},
  pages={600--609},
  year={2001},
  organization={IEEE}
}

@InProceedings{yung2024,
  author =	{Yung, Kingsley},
  title =	{{Limits of Sequential Local Algorithms on the Random k-XORSAT Problem}},
  booktitle =	{51st International Colloquium on Automata, Languages, and Programming (ICALP 2024)},
  pages =	{123:1--123:20},
  series =	{Leibniz International Proceedings in Informatics (LIPIcs)},
  year =	{2024},
  volume =	{297},
  publisher =	{Schloss Dagstuhl -- Leibniz-Zentrum f{\"u}r Informatik},
}

@article{hastad-clique,
author = {Johan H{\aa}stad},
title = {{Clique is hard to approximate within $n^{1-\epsilon}$}},
volume = {182},
journal = {Acta Mathematica},
number = {1},
publisher = {Institut Mittag-Leffler},
pages = {105 -- 142},
year = {1999},
}

@article{cooley2021large,
  title={Large induced matchings in random graphs},
  author={Cooley, Oliver and Draganic, Nemanja and Kang, Mihyun and Sudakov, Benny},
  journal={SIAM Journal on Discrete Mathematics},
  volume={35},
  number={1},
  pages={267--280},
  year={2021},
  publisher={SIAM}
}

@article{cameron1989induced,
  title={Induced matchings},
  author={Cameron, Kathie},
  journal={Discrete Applied Mathematics},
  volume={24},
  number={1-3},
  pages={97--102},
  year={1989},
  publisher={Elsevier}
}

@article{chakraborti2023extremal,
  title={Extremal bipartite independence number and balanced coloring},
  author={D. Chakraborti},
  journal={European Journal of Combinatorics},
  volume={113},
  pages={103750},
  year={2023},
  publisher={Elsevier}
}

@article{dhawan2025balanced,
  title={Balanced independent sets and colorings of hypergraphs},
  author={Dhawan, Abhishek},
  journal={Journal of Graph Theory},
  volume={109},
  number={1},
  pages={43--51},
  year={2025},
  publisher={Wiley Online Library}
}

@article{Grimmett_McDiarmid_1975, title={On colouring random graphs}, volume={77}, DOI={10.1017/S0305004100051124}, number={2}, journal={Mathematical Proceedings of the Cambridge Philosophical Society}, author={Grimmett, G. R. and McDiarmid, C. J. H.}, year={1975}, pages={313–324}}

@article{dhawan2025balancedER,
  title={Balanced colorings of Erd\H{o}s-R\'enyi hypergraphs},
  author={Dhawan, Abhishek and Wang, Yuzhou},
  journal={arXiv preprint arXiv:2504.04585},
  year={2025}
}

@article{feige2010balanced,
  title={Balanced coloring of bipartite graphs},
  author={U. Feige and S. Kogan},
  journal={Journal of Graph Theory},
  volume={64},
  number={4},
  pages={277--291},
  year={2010},
  publisher={Wiley Online Library}
}

\end{document}